\documentclass[11pt]{article}
\usepackage{lineno,hyperref}
\usepackage{bbm}
\usepackage{float}
\usepackage{mathrsfs}
\usepackage{graphicx}
\usepackage[utf8]{inputenc}
\usepackage{amsmath}
\usepackage{amstext}
\usepackage{amsfonts}
\usepackage{amsthm}
\usepackage{marvosym}
\usepackage{amssymb}
\usepackage[bf,SL,BF]{subfigure}
\usepackage{subfigure}
\usepackage{caption}
\newcommand{\triple}{{\vert\kern-0.25ex\vert\kern-0.25ex\vert}}
\DeclareUnicodeCharacter{2010}{-}
\theoremstyle{plain}
\allowdisplaybreaks
\newtheorem{definition}{Definition}[section]
\newtheorem{theorem}[definition]{Theorem}
\newtheorem{lemma}[definition]{Lemma}

\newtheorem{assumption}[definition]{Assumption}
\theoremstyle{definition}

\renewcommand\labelenumi{(\roman{enumi})}
\renewcommand\theenumi\labelenumi

\begin{document}
\title{\bf Numerical approximation of hybrid Poisson-jump Ait-Sahalia-type interest rate model with delay }
\author
{{\bf Emmanuel Coffie \footnote{Corresponding author, Email: emmanuel.coffie@strath.ac.uk}}
\\[0.2cm]
Department of Mathematics and Statistics,
 \\[0.2cm]
University of Strathclyde,  Glasgow G1 1XH, U.K.}
\date{}
\maketitle
\begin{abstract}
While the original Ait-Sahalia interest rate model has been found considerable use as a model for describing time series evolution of interest rates, it may not possess adequate specifications to explain responses of interest rates to empirical phenomena such as volatility 'skews' and 'smiles', jump behaviour, market regulatory lapses, economic crisis, financial clashes, political instability, among others collectively. The aim of this paper is to propose a modified version of this model by incorporating additional features to collectively describe these empirical phenomena adequately. Moreover, due to lack of a closed-form solution to the proposed model, we employ several new truncated EM techniques to examine this model and justify the scheme within Monte Carlo framework to compute expected payoffs of some financial quantities such as a bond and a barrier option.

\end{abstract}
\section{Introduction}
The shortcoming of the continuous-time model of Black-Scholes \cite{blackshole} in describing convex phenomena of implied volatility exhibited by most historical financial data led to the underlying assumption of constant volatility to be questioned. Several empirical studies have rather shown that stochastic volatility models with inherent features of past dependency are suitable models for describing convex phenomena of implied volatility against market anomalies (see, e.g., \cite{mao3, wu3, cev,ratanov}).
\par
It has also been well known that asset prices admit jumps in response to lack of information or unexpected catastrophic news.  This phenomenon typically generates price vibrations with larger quantiles than normal (see \cite{espen}). Apparently, this violates the efficient market hypothesis that all available information are reflected in current asset prices. There are several existing rich literature where the authors employed jump-diffusions models to describe jump behaviour of asset prices arising from lack of information or unexpected catastrophic news (see, e.g.,  \cite{jmerton,linyeh, kou, wu2}).
\par
Hybrid models driven by finite-state Markovian chains have also been increasingly employed as suitable models for modelling uncertainty in modern economic or financial systems (see, e.g.,\cite{Elliott, hamilton, ratanov, bollen}). The hybrid models randomly switch between finite number of regimes in anticipation to unexpected abrupt structural changes in underlying economic or financial mechanisms.
\par
The Ait-Sahalia interest rate model which is popularly used to describe time-series evolution of interest rates is driven by a strongly nonlinear stochastic differential equation
\begin{equation}\label{eq:1}
 dx(t)=(\alpha_{-1}x(t)^{-1}-\alpha_{0}+\alpha_{1}x(t)-\alpha_{2}x(t)^{2})dt+\sigma x(t)^{\theta}dB(t),
\end{equation}
where $\alpha_{-1},\alpha_{0},\alpha_{1}, \alpha_{2}>0$ and $\theta >1$. For more extensive existing literature concerning with SDE \eqref{eq:1}, the readers, for instance, may consult \cite{aitsahalia}, \cite{cheng}, \cite{Szpruch} and \cite{dung} among others. 
\par
Despite of the wide applicability of \textup{SDE} \eqref{eq:1}, this model may not possess inherent features to fully describe dynamical behaviours of interest rates in response to unexpected joint effects of extreme volatility, jumps, financial clashes, economic crisis among others. To help describe joint effects of these phenomena, we may specify \textup{SDE}\eqref{eq:1} as a hybrid Poisson-driven jump SDDE governed by
\begin{align}\label{eq:2}
\begin{cases}
  dx(t)=(\alpha_{-1}(r(t))x(t^-)^{-1}-\alpha_{0}(r(t))+\alpha_{1}(r(t))x(t^-)-\alpha_{2}(r(t))x(t^-)^{\rho})dt \\[0.2cm]
     \hspace*{2cm}+\varphi(x((t-\tau)^-),r(t))x(t^-)^{\theta}dB(t)+\alpha_{3}(r(t))x(t^-)dN(t), \hspace*{1cm} \text{$t\geq 0$},  \\[0.0005cm]
  x(t)=\xi(t), r(0)=r_0, \hspace*{0.2cm} \text{$t\in[-\tau, 0]$}.
\end{cases}
\end{align}
\\
Here $\rho, \theta >1$, $r(\cdot)$ is a Markov chain with finite space $\mathcal{S}=\{1,2,..,N\}$, $\alpha_{-1},\alpha_{0},\alpha_{1}, \alpha_{2}$ and $\alpha_{3}$ are functions of $r(\cdot)$, $\varphi(\cdot,\cdot)$ depends on $r(\cdot)$ and $x(t-\tau)$, $\tau>0$ and $x(t-\tau)$ denotes delay in $x(t)$. Moreover, $x(t^-)=\lim_{s\rightarrow t^-}x(s)$, $N(t)$ is a scalar Poisson process independent of a scalar  Brownian motion $B(t)$, with compensated Poisson process given by $\widetilde{N}(t)=N(t)-\lambda t$, where $\lambda$ is a jump intensity.
\par
The SDDE \eqref{eq:2} integrates three unique specifications under a unified framework. For instance, the delay in volatility function may capture the dynamical behaviours of implied volatility. On the other hand, the Poisson-driven term may explain tail distribution of interest rates in response to unexpected catastrophic news. The Markovian switching term may address effects of unpredictable market shocks which may arise from abrupt changes such as regulatory lapses, financial clashes, economic crisis, political instability or unobservable states of the underlying market frameworks or mechanisms.
\par
The solution to SDDE \eqref{eq:2} obviously cannot be found by closed-form formula. It is also obvious SDDE \eqref{eq:2} has  super-linear coefficient terms. As a result, we cannot employ the classical global Lipschitz-based techniques for numerical analysis of SDDE \eqref{eq:2}. To the best of our knowledge, there exists no relevant literature devoted to numerical analysis of system of SDDE \eqref{eq:2} in the strong sense. This therefore calls for a need to investigate feasibility of \textup{SDDE} \eqref{eq:2} from viewpoint of applications.
\par
In this work, we  will focus on developing several new truncated EM techniques to numerically study SDDE \eqref{eq:2}. The rest of the paper is organised as follows: In section 2, we  will examine the existence and uniqueness of the solution to SDDE \eqref{eq:2} and show that the solution will always be positive. We will also establish moment bounds of the exact solution in this section. In section 3, we will define the truncated EM scheme for \textup{SDDE} \eqref{eq:2} and survey moment bounds of the numerical solutions. We will employ truncated EM techniques to establish finite time strong convergence theory in section 4. in section 5, we will also implement some numerical examples to validate efficiency of the proposed scheme.  Finally, in the last section, we will justify the convergence result within a Monte Carlo scheme to value some financial products such a bond and a path-dependent barrier option.
\section{ Mathematical preliminaries}
\noindent Throughout this paper unless otherwise specified, we let $\{\Omega,\mathcal{F}, \{\mathcal{F}_t\}_{t\geq 0}, \mathbb{P} \}$ be a complete probability space with filtration $\{\mathcal{F}_t\}_{t\geq 0}$ satisfying the usual conditions (i.e, it is increasing and right continuous while $\mathcal{F}_0$ contains all $\mathbb{P}$-null sets). If $x, y$ are real numbers, then we denote $x\vee y=\max(x, y)$ and $x\wedge y=\min(x,y)$. For $\tau >0$, $C([-\tau,0];(0,\infty))$ denotes the space of all continuous functions $\xi: [-\tau,0]\rightarrow (0,\infty)$ with the norm $\|\xi\|=\sup_{-\tau\leq u\leq 0}\xi(u)$. Also let $\mathbb{R}_+=(0,\infty)$ and $C(\mathbb{R_+};\mathbb{R_+})$ denote the space of all nonnegative continuous functions defined on $\mathbb{R}_+$. Moreover, let $\emptyset$ denote the empty set so that $\text{inf } \emptyset=\infty$. For a set $A$, denote its indication function by $1_A$. For $t\geq 0$, let $B(t)$ be a scalar Brownian motion and $N(t)$ be a scalar Poisson process with jump intensity $\lambda$ which is independent of the Brownian motion, defined on the above probability space. Also let $r(t), t\geq 0$, be a right-continuous Markov chain defined on the above probability space taking values in a finite state space $\mathcal{S}=\{1,2,....,N\}$ with the generator $\Gamma=(\gamma_{ij})_{N\times N}$ given by
\begin{equation}\label{eq:3}
\mathbb{P}\{ r(t+\delta)=j|r(t)=i\}=
\begin{cases}
\gamma_{ij}\delta+o(\delta)& \mbox{if $i\neq j$},  \\
1+\gamma_{ij}\delta+o(\delta)& \mbox{if $i=j$},
\end{cases}
\end{equation}
where $\delta>0$. Here $\gamma_{ij}\geq 0$ is the transition rate from $i$ to $j$ if $j\neq i$ while
\begin{equation}\label{3*}
 \gamma_{ii}=-\sum_{i\neq j}\gamma_{ij}.
\end{equation}
We assume that the Markov chain  $r(\cdot)$ is ${\cal F}_t$-adapted but independent of the Brownian motion $B(\cdot)$ and Poisson process $N(\cdot)$. It is well known that almost every sample path of $r(\cdot)$ is a right-continuous step function with finite number of simple jumps in finite subinterval of $[0,\infty)$. Consider the following scalar dynamics as equation of \textup{SDDE} \eqref{eq:2}
\begin{align}\label{eq:4}
dx(t)&=f(x(t^-),r(t))dt+\varphi(x((t-\tau)^-)),r(t))g(x(t^-))dB(t)+h(x(t^-),r(t))dN(t),
\end{align}
such that $f(x,i)=\alpha_{-1}(i)x^{-1}-\alpha_{0}(i)+\alpha_{1}(i)x-\alpha_{2}(i)x^{\rho}$, $g(x)=x^{\theta}$, $h(x,i)=\alpha_{3}(i)x$, $ \forall x\in \mathbb{R}_+$ and $i\in S$, where $\varphi(\cdot,\cdot)\in C(\mathbb{R}_+\times \mathcal{S};\mathbb{R}_+)$. For each Lyapunov function $H\in C^{2,1}(\mathbb{R}\times \mathbb{R}_+\times \mathcal{S};\mathbb{R})$, define  the jump-diffusion operator $\mathcal{L}H:\mathbb{R}\times \mathbb{R}\times \mathbb{R}_+\times \mathcal{S}\rightarrow \mathbb{R}$ by
\begin{equation}\label{eq:5}
\mathcal{L}H(x,y,t,i)=\mathcal{I}H(x,y,t,i)+\lambda(H(x+h(x),t,i)-H(x,t,i))+\sum_{j=1}^{N}\gamma_{ij}H(x,t,j),
\end{equation}
where $\mathcal{I} H:\mathbb{R}\times \mathbb{R}\times \mathbb{R}_+\times \mathcal{S}\rightarrow\mathbb{R}$ is the diffusion operator defined by
\begin{equation}\label{eq:6}
\mathcal{I}H(x,y,t,i)=H_t(x,t,i)+H_x(x,t,i)f(x)+\frac{1}{2}H_{xx}(x,t,i)\varphi(y,i)^2g(x)^2,
\end{equation}
with $H_t(x,t,i)=\frac{\partial H(x,t,i)}{\partial t}$, $H_x(x,t,i)=\frac{\partial H(x,t,i)}{\partial x}$ and $H_{xx}(x,t,i)=\frac{\partial^2 H(x,t,i)}{\partial x^2}$. Given the jump-diffusion operator, we could deduce the generalised It\^{o} formula as
\begin{align}\label{eq:7}
dH(x(t),t,r(t))&=\mathcal{L}H(x(t^-),x((t-\tau)^-),t,r(t))dt\nonumber\\
&+H_x(x(t^-),t,r(t))\varphi(x((t-\tau)^-),r(t))g(x(t^-))dB(t)\nonumber\\
&+(H(x(t^-))+h(x(t^-)),t,r(t))-H(x(t^-),t,r(t)))d\tilde{N}(t)\nonumber\\
&+\int_{\mathbb{R}}(H(x(t^-),t,i_0+q(x(t^-),z))-H(x(t^-),t,r(t))))M(dt,dz), \quad \text{a.s}.
\end{align}
Consult \cite{{yumao}} and the references therein regarding the function $q(\cdot)$ and the martingale measure $M(\cdot,\cdot)$. We impose the following standing hypotheses which will be recalled later.
\begin{assumption}\label{eq:a1}
The volatility function $\varphi:\mathbb{R_+}\times S\rightarrow \mathbb{R_+}$ of \textup{SDDE} \eqref{eq:4} is Borel-measurable and bounded by a positive constant, that is
\begin{equation}\label{eq:8}
  \varphi(y,i)\leq \sigma,
\end{equation}
$\forall y\in \mathbb{R}_+ $ and $i\in \mathcal{S}$.
\end{assumption}
\begin{assumption}\label{eq:a3*}
For any $R>0$, there exists a constant $L_R> 0$ such that the volatility function $\varphi(\cdot,\cdot)$ of \textup{SDDE} \eqref{eq:4} satisfies
\begin{equation}\label{eq:13}
|\varphi(y,i)-\varphi(\bar{y},i)| \leq L_R |y-\bar{y}|,\quad
\end{equation}
$\forall (y,\bar{y})\in [1/\mathbb{R},\mathbb{R}]$ and $i\in \mathcal{S}$.
\end{assumption}
\begin{assumption}\label{eq:a2}
The parameters of \textup{SDDE} \eqref{eq:4} obey
\begin{equation}\label{eq:9}
1+ \rho> 2\theta,\quad \rho, \theta >1.
\end{equation}
\end{assumption}
\section{Analytical properties}
In this section, we study the existence of pathwise uniqueness and boundness of moments of the exact solution to  \textup{SDDE}  \eqref{eq:4}.
\subsection{Global positive solution}
One basic requirement of a financial model is the existence of a pathwise unique positive solution. The following lemma therefore reveals this requirement.
\begin{lemma}\label{eq:L0}
Let Assumptions \ref{eq:a1} and \ref{eq:a2} hold. Then for any given initial data
\begin{equation}\label{eq:10}
\{ x(t): -\tau\leq t \leq 0\}=\xi(t) \in C([-\tau,0]:\mathbb{R}_+), \quad r_0 \in \mathcal{S},
\end{equation}
there exists a unique global solution $x(t)$ to \textup{SDDE}  \eqref{eq:4} on $t\ge -\tau$ and $x(t)>0$ a.s.
\end{lemma}
\begin{proof} Since the coefficient terms of \textup{SDDE}  \eqref{eq:4} are locally Lipschitz continuous in $[-\tau,\infty)$, then there exists a unique positive maximal local solution $x(t)\in [-\tau,\tau_{e})$ for any given initial data \eqref{eq:10}, where $\tau_{e}$ is the explosion time (e.g., see \cite{yumao}). Let  $n_0>0$ be sufficiently large such that
\begin{equation*}\label{eq:11}
  \frac{1}{n_0}<\underset{-\tau\leq t\leq0}{\min}|\xi(t)|\leq\underset{-\tau\leq t\leq 0}{\max}|\xi(t)|<n_0.
\end{equation*}
For each integer $n\geq n_0$, define the stopping time
\begin{equation}\label{eq:12}
\tau_n=\inf\{ t\in [0,\tau_{e}):x(t)\not\in(1/n,n)\}.
\end{equation}
Obviously, $\tau_n$ is increasing  as $n\rightarrow \infty$. Set $\tau_{\infty}=\underset{n\rightarrow \infty}\lim \tau_n$, whence $\tau_{\infty}\leq \tau_e$ a.s. In other words, to complete the proof, we need to show that
\begin{equation*}
\tau_{\infty}=\infty\quad \text{a.s}. 
\end{equation*}
We define a $C^2$-function $H:\mathbb{R}_+\to \mathbb{R}_+$ for some $\phi\in(0,1]$ by
\begin{equation}
H(x)=x^{\phi}-1-\phi\text{log}(x).
\end{equation}
From the operator \eqref{eq:6} and by Assumption \ref{eq:a1}, we obtain 
\begin{align*}
\mathcal{I}H(x,y,t,i)&\le  \alpha_{-1}(i)\phi x^{\phi-2}-\alpha_0(i)\phi x^{\phi-1}+\alpha_1(i)\phi x^{\phi}-\alpha_2(i)\phi x^{\rho+\phi-1}-\alpha_{-1}(i)\phi x^{-2}\\
&+\alpha_0(i)\phi x^{-1}-\alpha_1(i)\phi+\alpha_2(i)\phi x^{\rho-1}+\frac{\sigma^2}{2}\phi(\phi-1)x^{\phi+2\theta-2}+\frac{\sigma^2}{2}\phi x^{2\theta-2}.
\end{align*}
By the Jump-diffusion operator in \eqref{eq:5}, we now have
\begin{align*}
\mathcal{L}H(x,y,t,i)\le \mathcal{I}H(x,y,t,i)+\lambda((1+\alpha_3(i))^{\phi}-1)x^{\phi}-\lambda \phi\log(1+\alpha_3(i)).
\end{align*}
For $\phi\in(0,1]$ and by Assumption \ref{eq:a2}, we observe $-\alpha_{-1}(i)\phi x^{-2}$ dominates and tends to $-\infty$ for small $x$ and for large $x$, $-\alpha_2(i)\phi x^{\rho+\phi-1}$ dominates and tends to $-\infty$. So there exists a constant $\mathcal{K}_0$ such that
\begin{equation*}
\mathcal{L}H(x,y,t,i)\le \mathcal{K}_0.
\end{equation*}
So for any arbitrary $t_1\ge 0$, the It\^{o} formula gives us
\begin{equation*}
\mathbb{E}[H(x(\tau_n\wedge t_1))]\leq H(\xi(0))+\mathcal{K}_0t_1.
\end{equation*}
It then follows
\begin{equation*}
\mathbb{P}(\tau_n\leq t_1)\leq\frac{H(\xi(0))+\mathcal{K}_0t_1}{H(1/n)\wedge H(n)}.
\end{equation*}
This implies $\mathbb{P}(\tau_{\infty}\le t_1)=0$ and consequently, we must have 
\begin{equation*}
\mathbb{P}(\tau_{\infty}=\infty)=1
\end{equation*}
as the required assertion. The proof is thus complete.
\end{proof}
\subsection{Moment boundedness }
\noindent The following lemma shows the moment of the exact solution $x(t)$ to \textup{SDDE}  \eqref{eq:3} is upper bounded.
\begin{lemma}\label{eq:L1}
Let Assumptions \ref{eq:a1} and \ref{eq:a2} hold. Then for any $p>2\vee(\rho-1)$, the solution $x(t)$ to \textup{SDDE}  \eqref{eq:3} satisfies
\begin{equation}\label{eq:13}
\sup_{0\leq t<\infty}\Big(\mathbb{E}|x(t)|^p\Big)\leq c_1
\end{equation}
and consequently
\begin{equation} \label{eq:14*}
\sup_{0\leq t<\infty}\Big(\mathbb{E}|\frac{1}{x(t)}|^p\Big)\leq c_2,
\end{equation}
where $c_1$ and $c_2$ are constants.
\end{lemma}
\begin{proof}
For every sufficiently large integer $n$, we define the stopping time  by
\begin{equation*}
\tau_n=\inf \{ t\geq 0:x(t)\not\in(1/n,n)\}.
\end{equation*}
We also define a Lyapunov function $H\in C^{2,1}(\mathbb{R}_+\times \mathbb{R}_+;\mathbb{R}_+)$ by $H(x,t)=e^tx^p$ . By Assumption \ref{eq:a1}, we apply \eqref{eq:4} to obtain
\begin{equation*}
 \mathcal{L}H(x,y,t,i) \leq  \mathcal{I} H(x,y,t,i)+\lambda e^tx^p[(1+\alpha_{3}(i))^p-1],
\end{equation*}
where
\begin{equation*}
\mathcal{I} H(x,y,t,i)\leq e^t\Big[x^p+px^{p-2}(\alpha_{-1}(i)-\alpha_{0}(i)x+\alpha_{1}(i)x^2-\alpha_{2}(i)x^{\rho+1}+\frac{(p-1)}{2}\sigma^2x^{2\theta})\Big].
\end{equation*}
Apparently, by Assumption \ref{eq:a2}, $-p\alpha_{2}(i)x^{\rho+p-1}$ dominates and tends to $-\infty$ for large $x$. So there exists a constant $\mathcal{K}_1 $ such that
\begin{equation*}
  \mathcal{L}H(x,y,t,i)\leq \mathcal{K}_1 e^t.
\end{equation*}
By the the It\^{o} formula, we have
\begin{equation*}
\mathbb{E}[e^{t\wedge \tau_n}|x(t\wedge \tau_n)|^p]\leq |\xi(0)|^p+\mathcal{K}_1 e^t.
\end{equation*}
Applying  the Fatou lemma and letting $n\rightarrow \infty$ yields
\begin{equation*}
\mathbb{E}|x(t)|^p<\frac{|\xi(0)|^p+\mathcal{K}_1 e^t}{e^t}
\end{equation*}
and consequently, we obtain \eqref{eq:16} as the required assertion. Moreover, by applying the operator \eqref{eq:4} to the Lyapunov function $H(x,t)=e^t/x^p$, we compute
\begin{equation*}
 \mathcal{L}H(x,y,t,i) \leq  \mathcal{I} H(x,y,t,i)+\lambda e^tx^{-p}[(1+\alpha_{3}(i))^{-p}-1],
\end{equation*}
where  Assumption \ref{eq:a1} has been used and
\begin{equation*}
\mathcal{I} H(x,y,t,i)\leq e^t\Big[x^{-p}-px^{-(p+2)}(\alpha_{-1}(i)-\alpha_{0}(i)x+\alpha_{1}(i)x^2-\alpha_{2}(i)x^{\rho+1}+\frac{(p+1)}{2}\sigma^2x^{2\theta})\Big].
\end{equation*}
For $p>2\vee(\rho-1)$, we note $-\alpha_{-1}(i)px^{-(p+2)}$ dominates and tends to $-\infty$  for small $x$. Moreover, we also note $p\alpha_{2}(i)x^{\rho-p-1}$ dominates and tends to 0 for large $x$. We then find a constant $\mathcal{K}_2$ such that
\begin{equation*}
  \mathcal{L}H(x,y,t,i)\leq \mathcal{K}_2e^t.
\end{equation*}
So from the It\^{o} formula, we can apply the Fatou lemma and let $n\rightarrow \infty$  to arrive at \eqref{eq:14*}.
\end{proof}
\section{Numerical method}
Under this section, we recall the truncated EM method and apply it for convergent approximation of \textup{SDDE} \eqref{eq:4}. To start with, let also impose the following useful condition on the initial data.
\begin{assumption}\label{eq:a3}
There is a pair of constant $\mathcal{K}_3>0$ and $\Upsilon\in (0,1]$ such that for all $ -\tau\leq s \leq t \leq 0$, the initial data $\xi$ satisfies
\begin{equation}\label{eq:15}
  |\xi(t)-\xi(s)|\leq \mathcal{K}_3|t-s|^{\Upsilon}.
\end{equation}
\end{assumption}
We also need the following lemmas (see \cite{emma}).
\begin{lemma}\label{eq:L2}
For any $R>0$, there exists a constant $K_R > 0$ such that the coefficient terms of \textup{SDDE} \eqref{eq:4} satisfy
\begin{equation}\label{eq:16}
|f(x,i)-f(\bar{x},i)|\vee |g(x)-g(\bar{x})|\vee |h(x,i)-h(\bar{x},i)| \leq K_R|x-\bar{x}|, 
\end{equation}
$\forall x,\bar{x} \in [1/\mathbb{R},\mathbb{R}]\text{ and } i\in \mathcal{S}$.
\end{lemma}
\begin{lemma}\label{eq:L3}
Let Assumptions \ref{eq:a1} and \ref{eq:a2} hold. For any $p\geq 2$, there exists $\mathcal{K}_4=\mathcal{K}_4(p)>0$ such that the coefficient terms of \textup{SDDE}  \eqref{eq:4} satisfy
\begin{equation}\label{eq:14}
xf(x,i)+\frac{p-1}{2}|\varphi(y,i)g(x)|^2\leq \mathcal{K}_4(1+|x|^2),
\end{equation}
$\forall(x,y,i)\in \mathbb{R}_+\times\mathbb{R}_+\times \mathcal{S}$.
\end{lemma}
The truncated EM scheme for SDDE \eqref{eq:4} is now defined in the following subsection.
\subsection{The truncated EM method}
Let extend the volatility function $\varphi(\cdot,\cdot)$ and the jump term $h(\cdot,\cdot)$ from $\mathbb{R}_+$ to $\mathbb{R}$ by setting $\varphi(y,i)=\varphi(0,i)$ and $h(x,i)=0$ for $x<0$. These extensions do not in any way affect above conditions and results. To define the truncated EM scheme for \textup{SDDE}  \eqref{eq:4}, we first choose a strictly increasing continuous function $\mu: \mathbb{R}_+ \rightarrow \mathbb{R}_+$ such that $\mu(r)\rightarrow \infty$ as $r\rightarrow \infty$ and
\begin{equation}\label{eq:17}
\sup_{1/r \leq x\leq r}(|f(x,i)|\vee g(x))\leq \mu(r), \quad \forall r> 1.
\end{equation}
Let $\mu^{-1}$ be the inverse function of $\mu$ and $\psi:(0,1)\rightarrow \mathbb{R}_+$ a strictly decreasing function such that
\begin{equation}\label{eq:18}
\quad \lim_{\Delta \rightarrow 0}\psi(\Delta)=\infty \text{  and  }\Delta^{1/4}\psi(\Delta)\leq 1, \quad \forall \Delta \in (0,1].
\end{equation}
Find $\Delta^*\in (0,1)$ such that $\mu^{-1}(\psi(\Delta^*))>1$ and $f(x,i)>0$ for $0<x<\Delta^*$. For a given step size $\Delta \in (0,\Delta^*)$, let us define the truncated functions
\begin{equation*}
f_{\Delta}(x,i)=f\Big(1/\mu^{-1}(\psi(\Delta))\vee (x\wedge \mu^{-1}(\psi(\Delta))),i\Big), \quad \forall (x,i) \in \mathbb{R}\times \mathcal{S}
\end{equation*}
and
\begin{equation*}
g_{\Delta}(x)=
\begin{cases}
  g\Big(x\wedge \mu^{-1}(\psi(\Delta))\Big), & \mbox{if $x\geq 0$ }\\
  0,                             & \mbox{if $x< 0$}.
\end{cases}
\end{equation*}
Then for $x\in [1/\mu^{-1}(\psi(\Delta)),  \mu^{-1}(\psi(\Delta))]$, we get
\begin{align*}
|f_{\Delta}(x,i)|&=|f(x,i)|\leq \underset{1/\mu^{-1}(\psi(\Delta))\leq z \leq\mu^{-1}(\psi(\Delta))}{\max |f(z,i)|}\\
             &\leq \mu(\mu^{-1}(\psi(\Delta)))= \psi(\Delta)
\end{align*}
and
\begin{align*}
g_{\Delta}(x)\leq \mu(\mu^{-1}(\psi(\Delta)))=\psi(\Delta).
\end{align*}
We easily see that
\begin{equation}\label{eq:19}
|f_{\Delta}(x,i)|\vee g_{\Delta}(x)\leq \psi(\Delta), \quad \forall(x,i) \in \mathbb{R}\times \mathcal{S}.
\end{equation}\noindent
The following lemma confirms $f_{\Delta}$ and $g_{\Delta}$ nicely reproduce \eqref{eq:14}. 
\begin{lemma}\label{eq:L4}
Let Assumption \ref{eq:a1} and \ref{eq:a2} hold. Then, for all $\Delta \in (0,\Delta^*)$ and $p\geq 2$, the truncated functions satisfy
\begin{equation}\label{eq:20}
xf_{\Delta}(x,i)+\frac{p-1}{2}|\varphi(y,i)g_{\Delta}(x)|^2\leq \mathcal{K}_5(1+|x|^2)
\end{equation}
$\forall (x,y,i)\in \mathbb{R}\times \mathbb{R}\times \mathcal{S}$, where $\mathcal{K}_5$ is independent of $\Delta$. Consult \cite{emma}. 
\end{lemma}
Let also recall the following useful lemma.
\begin{lemma}\label{eq:L5}
Given $\Delta>0$, let $r_{\Delta}^k=r_{\Delta}(k\Delta)$ for $k\ge 0$. Then $\{r_{\Delta}^k,k=0,1,2,...\}$ is a discrete Markov chain with the one-step transition probability matrix
\begin{equation*}
P(\Delta)=(P_{ij}(\Delta))_{N\times N}=e^{\Delta \Gamma}.
\end{equation*}
\end{lemma}
The discrete Markovian chain $\{r_{\Delta}^k,k=0,1,2,...\}$ can be simulated as follows: compute the one-step transition probability matrix
\begin{equation*}
P(\Delta)=(P_{ij}(\Delta))_{N\times N}=e^{\Delta \Gamma}.
\end{equation*}
Let $r_{\Delta}^0=i_0$ and generate a random number $\varpi$ which is uniformly distributed in $[0,1]$. Define
\begin{equation*}
r_{\Delta}^1=
\begin{cases}
  i_1 & \mbox{if $i_1\in \mathcal{S}-\{N\}$ such that $\sum_{j=i}^{i_1-1}P_{i_0,j(\Delta)}\le \varpi_1 <\sum_{j=i}^{i_1}P_{i_0,j(\Delta)} $ }\\
  N, & \mbox{if $\sum_{j=i}^{N-1}P_{i_0,j(\Delta)}\le \varpi_1$},
\end{cases}
\end{equation*}
where we set $\sum_{j=i}^{0}P_{i_0,j(\Delta)}=0$ as usual. Generate independently a new random number $\varpi_2$ which is again uniformly distributed in $[0,1]$ and then define 
\begin{equation*}
r_{\Delta}^2=
\begin{cases}
  i_2 & \mbox{if $ i_2\in \mathcal{S}-\{N\}$ such that $\sum_{j=i}^{i_2-1}P_{r_{\Delta}^1,j(\Delta)}\le \varpi_2 <\sum_{j=i}^{i_2}P_{r_{\Delta}^1,j(\Delta)} $ }\\
  N, & \mbox{if $\sum_{j=i}^{N-1}P_{r_{\Delta}^1,j(\Delta)}\le \varpi_2$},
\end{cases}
\end{equation*}
Repeating this procedure, a trajectory of $\{r_{\Delta}^k,k=0,1,2,...\}$ can be generated.
\par
Given the discrete Markovian chain scheme, we now form the discrete-time truncated EM scheme for \textup{SDDE} \eqref{eq:4} by first letting $T>0$ be arbitrarily fixed and the step size $\Delta\in (0,\Delta^*]$ be a fraction of $\tau$. Define $\Delta=\tau/M$ for some positive integer $M$. Define $t_k=k\Delta$ for $k=-M,-(M-1),..,0,1,2,..$, set $X_{\Delta}(t_k)= \xi(t_k)$ for $k=-M,-(M-1),..,0$  and then compute
\begin{align}\label{eq:21}
X_{\Delta}(t_{k+1})&=X_{\Delta}(t_k)+f_{\Delta}(X_{\Delta}(t_k),r_{\Delta}(t_k))\Delta+\varphi(X_{\Delta}(t_{k-M}),r_{\Delta}(t_k))g_{\Delta}(X_{\Delta}(t_{k}))\Delta B_k\nonumber\\
&+h(X_{\Delta}(t_{k}),r_{\Delta}(t_k))\Delta N_k
\end{align}
for $k=0,1,2...,$ where $\Delta B_k=B(t_{k+1})-B(t_k)$ and $\Delta N_k=N(t_{k+1})-N(t_k)$. We have two versions of the continuous-time truncated EM solutions. The first one is defined by
\begin{equation}\label{eq:22}
\bar{x}_{\Delta}(t)=\sum_{k=-M}^{\infty}X_{\Delta}(t_k)1_{[t_k,t_{k+1})}(t) \text{ and }\bar{r}_{\Delta}(t)=\sum_{k=-M}^{\infty}r_{\Delta}(t_k)1_{[t_k,t_{k+1})}(t).
\end{equation}
These are the continuous-time step processes $\bar{x}_{\Delta}(t)$ and $\bar{r}_{\Delta}(t)$ on $t\ge -\tau$, where $1_{[t_k,t_{k+1})}$ is the indicator function on $[t_k,t_{k+1})$. The second one is the continuous-time continuous process $x_{\Delta}(t)$ on $t\ge -\tau$ defined by setting $x_{\Delta}(t)=\xi(t)$ for $t\in [-\tau,0]$ while for $t\geq 0$
\begin{align}\label{eq:23}
x_{\Delta}(t)&=\xi(0)+\int_0^t f_{\Delta}(\bar{x}_{\Delta}(s^-),\bar{r}_{\Delta}(s))ds+\int_0^t \varphi(\bar{x}_{\Delta}((s-\tau)^-),\bar{r}_{\Delta}(s))g_{\Delta}(\bar{x}_{\Delta}(s^-))dB(s)\nonumber\\
&+\int_0^t h(\bar{x}_{\Delta}(s^-),\bar{r}_{\Delta}(s))dN(s).
\end{align}
Apparently $x_{\Delta}(t)$ is an It\^{o} process on $t\geq 0$ satisfying It\^{o} differential
\begin{align}\label{eq:24}
dx_{\Delta}(t)&= f_{\Delta}(\bar{x}_{\Delta}(t^-),\bar{r}_{\Delta}(t))dt+\varphi(\bar{x}_{\Delta}((t-\tau)^-),\bar{r}_{\Delta}(t))g_{\Delta}(\bar{x}_{\Delta}(t^-))dB(t)\nonumber\\
&+h(\bar{x}_{\Delta}(t^-),\bar{r}_{\Delta}(t))dN(t).
\end{align}
\noindent We observe $x_{\Delta}(t_{k})=\bar{x}_{\Delta}(t_k)=X_{\Delta}(t_k)$, for all $k=-M,-(M-1),....$.
\section{Numerical properties}
Let now investigate the numerical properties of the truncated EM scheme. In the sequel, we let
\begin{equation*}
k(t)=[t/\Delta]\Delta,
\end{equation*}
for any $t\in [0,T]$, where $[t/\Delta]$ denotes the integer part of $t/\Delta$. The following lemma affirms $x_{\Delta}(t)$ and $\bar{x}_{\Delta}(t)$ are close to each other in strong sense.
\subsection{Moment boundedness}
\begin{lemma}\label{eq:L6}
Let Assumption \ref{eq:a1} hold. Then for any fixed $\Delta\in (0,\Delta^*]$, we have for $p\in[2,\infty)$
\begin{equation}\label{eq:25}
 \mathbb{E}\Big(|x_{\Delta}(t)-\bar{x}_{\Delta}(t)|^{p}\big| \mathcal{F}_{k(t)}\Big)\leq \mathfrak{C}_1\Big(\Delta^{p/2}(\psi(\Delta))^{p}+\Delta\Big)|\bar{x}_{\Delta}(t)|^{p}
\end{equation}
$\forall t\geq 0$, where $\mathfrak{C}_1$ denotes positive generic constants dependent only on $p$ and may change between occurrences.
\end{lemma}
\begin{proof}
Fix any $\Delta \in (0,\Delta^*)$ and $t\in [0,T]$. Then for $ p\in[2,\infty)$, we derive 
\begin{align*}
\mathbb{E}\Big(|x_{\Delta}(t)-\bar{x}_{\Delta}(t)|^{p}\big| \mathcal{F}_{k(t)}\Big)&\le 3^{p-1}\Big(\mathbb{E}\big(| \int_{k(t)}^{t} f_{\Delta}(\bar{x}_{\Delta}(s),\bar{r}_{\Delta}(s))ds|^{p}\big| \mathcal{F}_{k(t)}\big)\\
&+\mathbb{E}\big(|\int_{k(t)}^t \varphi(\bar{x}_{\Delta}((s-\tau)),\bar{r}_{\Delta}(s))g_{\Delta}(\bar{x}_{\Delta}(s))dB(s)|^{p}\big| \mathcal{F}_{k(t)}\big)\\
&+\mathbb{E}\big(|\int_{k(t)}^t h(\bar{x}_{\Delta}(s),\bar{r}_{\Delta}(s))dN(s)|^{p}\big| \mathcal{F}_{k(t)}\big) \Big)\\
&\leq 3^{p-1} \Big(\Delta^{p-1}\mathbb{E}( \int_{k(t)}^t |f_{\Delta}(\bar{x}_{\Delta}(s),\bar{r}_{\Delta}(s))|^{p} ds\big| \mathcal{F}_{k(t)})\\
&+ C_p\Delta^{(p-2)/2} \mathbb{E}(\int_{k(t)}^{t}|\varphi(\bar{x}_{\Delta}((s-\tau)),\bar{r}_{\Delta}(s))g_{\Delta}(\bar{x}_{\Delta}(s))|^{p} ds\big| \mathcal{F}_{k(t)})\\
&+\mathbb{E}(|\int_{k(t)}^t h(\bar{x}_{\Delta}(s),\bar{r}_{\Delta}(s))dN(s)|^{p}\big| \mathcal{F}_{k(t)} )\Big)\\
&\leq 3^{p-1}\Big(\Delta^{p-1}\Delta (\psi(\Delta))^{p} +C_p\Delta^{(p-2)/2}\Delta (\sigma \psi(\Delta))^{p}\\
&+\mathbb{E}(|\int_{k(t)}^t h(\bar{x}_{\Delta}(s),\bar{r}_{\Delta}(s))dN(s)|^{p}\big|
\mathcal{F}_{k(t)}) \Big).
\end{align*}
Recalling the characteristic function's argument $\mathbb{E}|\Delta N_k|^{p} \leq \bar{C}\Delta$, $\forall \Delta\in (0,\Delta^*)$, in \cite{HiKlo}, we note
\begin{align*}
\mathbb{E}\Big(|x_{\Delta}(t)-\bar{x}_{\Delta}(t)|^{p}\big| \mathcal{F}_{k(t)}\Big)&\le 3^{p-1}\Big(\Delta^{p-1}\Delta (\psi(\Delta))^{p} +C_p\Delta^{(p-2)/2}\Delta (\sigma \psi(\Delta))^{p}
+|h(\bar{x}_{\Delta}(t),r(t))|^{p} \mathbb{E}|\Delta N_k|^{p}\Big)\\
&\le 3^{p-1}\Big(\Delta^{p-1}\Delta (\psi(\Delta))^{p} +C_p\Delta^{(p-2)/2}\Delta (\sigma \psi(\Delta))^{p}+\bar{C} \alpha_3(i)^{p}|\bar{x}_{\Delta}(t)|^{p}\Delta\Big),
\end{align*}
where $h(\cdot,\cdot)$ and $\bar{C}>0$ are independent of $N_k$ and $\Delta$ respectively. We now have
\begin{align*}
\mathbb{E}\Big(|x_{\Delta}(t)-\bar{x}_{\Delta}(t)|^{p}\big| \mathcal{F}_{k(t)}\Big)&\le 3^{p-1}(1\vee C_p\sigma^{p}\vee \bar{C}\alpha_3(i)^{p})\Big(\Delta^{p/2}(\psi(\Delta))^{p}+|\bar{x}_{\Delta}(t)|^{p}\Delta \Big)\\
&\leq \mathfrak{C}_1\Big(\Delta^{p/2}(\psi(\Delta))^{p}+\Delta \Big)|\bar{x}_{\Delta}(t)|^{p},
\end{align*}
where 
$$\mathfrak{C}_1=3^{p-1}(1\vee C_p\sigma^{p}\vee \bar{C}\alpha_3^{p}) \text{ and } \alpha_3=\max_{i\in \mathcal{S}}\alpha_3(i).$$
Moreover, for $p\in(0,2)$, we obtain from the Jensen inequality that
\begin{align}\label{eq:le2}
\mathbb{E}\Big(|x_{\Delta}(t)-\bar{x}_{\Delta}(t)|^{p}\big| \mathcal{F}_{k(t)}\Big)&\le\Big[\mathfrak{C}_1\Big(\Delta(\psi(\Delta))^{2}+\Delta \Big)|\bar{x}_{\Delta}(t)|^{p} \Big]^{p/2}\nonumber\\
&\le 2^{p/2-1}\mathfrak{C}_1^{p/2}\Big(\Delta^{p/2}(\psi(\Delta))^{p}+\Delta^{p/2} \Big)(|\bar{x}_{\Delta}(t)|^{p})^{p/2}\nonumber\\
&\le \mathfrak{C}_2\Big(\Delta^{p/2}(\psi(\Delta))^{p}\Big)|\bar{x}_{\Delta}(t)|^{p},
\end{align}
where $\mathfrak{C}_2=2^{p/2}\mathfrak{C}_1^{p/2}$. The proof is complete.
\end{proof}
The following lemma reveals the numerical solutions have upper bound.
\begin{lemma}\label{eq:L7}
Let Assumptions \ref{eq:a1} and \ref{eq:a2} hold. Then for any $ p\ge 3$
\begin{equation}\label{eq:26}
  \sup_{0\leq \Delta \leq \Delta^*} \sup_{0\leq t\leq T}(\mathbb{E}|x_{\Delta}(t)|^p)\leq c_3, \quad \forall T> 0,
\end{equation}
where $c_3:=c_3(T, p, K, \xi)$ and may change between occurrences.
\end{lemma}
\begin{proof}
Fix any $\Delta \in (0,\Delta^*)$ and $T\geq 0$. For $t\in[0,T]$, we obtain from  \eqref{eq:5} and Lemma \ref{eq:L4} 
\begin{align*}
\mathbb{E}|x_{\Delta}(t)|^p-|\xi(0)|^p &\leq\mathbb{E}\int_{0}^{t}p|x_{\Delta}(s^-)|^{p-2}\Big(\bar{x}_{\Delta}(s^-)f_{\Delta}(\bar{x}_{\Delta}(s^-),\bar{r}_{\Delta}(s))\\
&+\frac{p-1}{2}|\varphi(\bar{x}_{\Delta}((s-\tau)^-),\bar{r}_{\Delta}(s))g_{\Delta}(\bar{x}_{\Delta}(s^-))|^2\Big)ds\\
&+\mathbb{E}\int_{0}^{t}p|x_{\Delta}(s^-)|^{p-2}(x_{\Delta}(s^-)-\bar{x}_{\Delta}(s^-))f_{\Delta}(\bar{x}_{\Delta}(s^-),\bar{r}_{\Delta}(s))ds\\
&+\lambda\mathbb{E}\Big(\int_{0}^{t}|x_{\Delta}(s^-)+h(\bar{x}_{\Delta}(s^-),\bar{r}_{\Delta}(s))|^p-|x_{\Delta}(s^-)|^p\Big)ds\\
&\leq \mathcal{J}_{11}+\mathcal{J}_{12}+\mathcal{J}_{13},
\end{align*}
where
\begin{align*}
\mathcal{J}_{11}&=\mathbb{E}\int_{0}^{t}\mathcal{K}_5p|x_{\Delta}(s^-)|^{p-2}(1+|\bar{x}_{\Delta}(s^-)|^2)ds\\
\mathcal{J}_{12}&=\mathbb{E}\int_{0}^{t}p|x_{\Delta}(s^-)|^{p-2}\Big(x_{\Delta}(s^-)-\bar{x}_{\Delta}(s^-)\Big)f_{\Delta}(\bar{x}_{\Delta}(s^-),\bar{r}_{\Delta}(s))ds\\
\mathcal{J}_{13}&=\lambda\mathbb{E}\Big(\int_{0}^{t}|x_{\Delta}(s^-)+h(\bar{x}_{\Delta}(s^-),\bar{r}_{\Delta}(s))|^p-|x_{\Delta}(s^-)|^p\Big)ds.
\end{align*}
The Young inequality gives us
\begin{align*}
\mathcal{J}_{11}&\leq \mathcal{K}_5\int_{0}^{t}\Big((p-2)\mathbb{E}|x_{\Delta}(s^-)|^{p}+2^p(1 +\mathbb{E}|\bar{x}_{\Delta}(s^-)|^{p})\Big)ds\\
&\leq \nu_1\int_{0}^{t}(1+\mathbb{E}|x_{\Delta}(s)|^{p}+\mathbb{E}|\bar{x}_{\Delta}(s)|^{p})ds,
\end{align*}
where $\nu_1=\mathcal{K}_5(2^p\vee (p-2))$. By the triangle inequality, we have for $p\ge 3$
\begin{align*}
\mathcal{J}_{12}&\le p\mathbb{E}\int_{0}^{t}\Big(|x_{\Delta}(s^-)-\bar{x}_{\Delta}(s^-)|+|\bar{x}_{\Delta}(s^-)|\Big)^{p-2}|x_{\Delta}(s^-)-\bar{x}_{\Delta}(s^-)||f_{\Delta}(\bar{x}_{\Delta}(s^-),\bar{r}_{\Delta}(s))|ds\\
&\le 2^{(p-3)}p\mathbb{E}\int_{0}^{t}\Big(|x_{\Delta}(s^-)-\bar{x}_{\Delta}(s^-)|^{p-2}+|\bar{x}_{\Delta}(s^-)|^{p-2}\Big)|x_{\Delta}(s^-)-\bar{x}_{\Delta}(s^-)||f_{\Delta}(\bar{x}_{\Delta}(s^-),\bar{r}_{\Delta}(s))|ds\\
&= \mathcal{J}_{121}+\mathcal{J}_{122},
\end{align*}
where 
\begin{align*}
\mathcal{J}_{121}&= 2^{(p-3)}p\mathbb{E}\int_{0}^{t}|\bar{x}_{\Delta}(s^-)|^{p-2}|x_{\Delta}(s^-)-\bar{x}_{\Delta}(s)||f_{\Delta}(\bar{x}_{\Delta}(s^-),\bar{r}_{\Delta}(s))|ds\\
\mathcal{J}_{122}&= 2^{(p-3)}p\mathbb{E}\int_{0}^{t}|x_{\Delta}(s^-)-\bar{x}_{\Delta}(s^-)|^{p-1}|f_{\Delta}(\bar{x}_{\Delta}(s^-),\bar{r}_{\Delta}(s))|ds.
\end{align*}
We now obtain from \eqref{eq:19} and \eqref{eq:le2}
\begin{align}\label{eq:h0}
\mathcal{J}_{121}&\le 2^{(p-3)}p\int_{0}^{t}\mathbb{E}\Big\{|\bar{x}_{\Delta}(s)|^{p-2}|f_{\Delta}(\bar{x}_{\Delta}(s),\bar{r}_{\Delta}(s))|\mathbb{E}\Big( |x_{\Delta}(s)-\bar{x}_{\Delta}(s)|\mathcal{F}_{k(s)})\Big)\Big\}ds\nonumber \\
&\le 2^{(p-3)}p\mathfrak{C}_2(\psi(\Delta))\Delta^{1/2}(\psi(\Delta))\int_{0}^{t}\mathbb{E}\Big\{|\bar{x}_{\Delta}(s)|(|\bar{x}_{\Delta}(s)|^{p-2})\Big\}ds\nonumber \\
&\le 2^{(p-3)}p\mathfrak{C}_2(\psi(\Delta))\Delta^{1/2}(\psi(\Delta))\int_{0}^{t}\mathbb{E}|\bar{x}_{\Delta}(s)|^{p-1}ds\nonumber \\
&\le 2^{(p-3)}\mathfrak{C}_2(\psi(\Delta))^2\Delta^{1/2}\int_{0}^{t}\Big(1+(p-1)\mathbb{E}|\bar{x}_{\Delta}(s)|^{p}\Big)ds\nonumber\\
&\le \nu_2+\nu_3\int_{0}^{t}\mathbb{E}|\bar{x}_{\Delta}(s)|^{p}ds,
\end{align}
where $\nu_2=2^{(p-3)}\mathfrak{C}_2T$, $\nu_3=2^{(p-3)}\mathfrak{C}_2(p-1)$ and $[(\psi(\Delta))\Delta^{1/4}]^2\le 1$.  We also have from \eqref{eq:19}
\begin{align}\label{eq:h1}
\mathcal{J}_{122}&\le 2^{(p-3)}p\psi(\Delta) \int_{0}^{t}\mathbb{E}|x_{\Delta}(s)-\bar{x}_{\Delta}(s)|^{p-1}ds.
\end{align}
We clearly observe that for $p\ge 3$ and $\varkappa\in(0,1/4]$, $p\varkappa\le (p-1)/2$ and hence
\begin{equation}\label{eq:h2}
\Delta^{(p-1)/2-\varkappa p}\le 1.
\end{equation}
So for $p\ge 3$ and $\varkappa=1/4$, we obtain from \eqref{eq:h1}, Lemma \ref{eq:L6}, \eqref{eq:h2} and the Young's inequality
\begin{align*}
\mathcal{J}_{122}&\le 2^{(p-3)}p\mathfrak{C}_1 \Big(\Delta^{(p-1)/2}(\psi(\Delta))^{p-1}(\psi(\Delta))+\Delta(\psi(\Delta))\Big)\int_{0}^{t}\mathbb{E}|\bar{x}_{\Delta}(s)|^{p-1}ds\\
&\le  2^{(p-3)}p\mathfrak{C}_1\Big(\Delta^{(p-1)/2}(\psi(\Delta))^{p}+\Delta(\psi(\Delta))\Big)\int_{0}^{t}\mathbb{E}|\bar{x}_{\Delta}(s)|^{p-1}ds\\
&\le 2^{(p-3)}p\mathfrak{C}_1\Big(\Delta^{(p-2)/4}+\Delta(\psi(\Delta))\Big)\int_{0}^{t}\mathbb{E}|\bar{x}_{\Delta}(s)|^{p-1}ds\\
&\le 2^{(p-2)}\mathfrak{C}_1\int_{0}^{t}\Big(1+(p-1)\mathbb{E}|\bar{x}_{\Delta}(s)|^{p}\Big)ds\\
&\le \nu_4+ \nu_5\int_{0}^{t}\mathbb{E}|\bar{x}_{\Delta}(s)|^{p}ds,
\end{align*}
where $\nu_4=2^{(p-2)}\mathfrak{C}_1T$ and $\nu_5=2^{(p-2)}\mathfrak{C}_1(p-1)$. We now combine $\mathcal{J}_{121}$ and $\mathcal{J}_{122}$ to have
\begin{align*}
\mathcal{J}_{12}&\le \nu_2+\nu_4+(\nu_3 +\nu_5)\int_{0}^{t}\mathbb{E}|\bar{x}_{\Delta}(s)|^{p}ds\\
&\le \nu_6+\nu_7\int_{0}^{t}\mathbb{E}|\bar{x}_{\Delta}(s)|^{p}ds,
\end{align*}
where $\nu_6=\nu_2+\nu_4$ and $\nu_7=\nu_3+\nu_5$.  Also we estimate $\mathcal{J}_{13}$ as
\begin{align*}
\mathcal{J}_{13}&=\lambda\mathbb{E}\Big(\int_{0}^{t}|x_{\Delta}(s^-)+h(\bar{x}_{\Delta}(s^-),\bar{r}_{\Delta}(s))|^p-|x_{\Delta}(s^-)|^p\Big)ds\\
   &\leq \lambda\mathbb{E}\Big(\int_{0}^{t}2^{p-1}|x_{\Delta}(s^-)|^p+2^{p-1}|h(\bar{x}_{\Delta}(s^-),\bar{r}_{\Delta}(s))|^p-|x_{\Delta}(s^-)|^p\Big)ds\\
   &\leq \lambda\mathbb{E}\Big(\int_{0}^{t}(2^{p-1}-1)|x_{\Delta}(s^-)|^p+2^{p-1}\alpha_3(i)^p|\bar{x}_{\Delta}(s^-)|^p\Big)ds\\
   &\leq \nu_8\int_{0}^{t}(\mathbb{E}|x_{\Delta}(s)|^p+\mathbb{E}|\bar{x}_{\Delta}(s)|^p)ds,
\end{align*}
where $\nu_8=\lambda((2^{p-1}-1)\vee 2^{p-1}\alpha_3^p)$ and $\alpha_3=\max_{i\in \mathcal{S}}\alpha_3(i)$. We combine $\mathcal{J}_{11}$, $\mathcal{J}_{12}$ and $\mathcal{J}_{13}$ to have
\begin{align*}
\mathbb{E}|x_{\Delta}(t)|^p &\le |\xi(0)|^p+(\nu_1T+\nu_6)+\int_{0}^{t}\Big((\nu_1+\nu_8)\mathbb{E}|x_{\Delta}(s)|^{p}+(\nu_1+\nu_7+\nu_8)\mathbb{E}|\bar{x}_{\Delta}(s)|^{p}\Big)ds\\
&\leq \nu_9+2\nu_{10}\int_{0}^{t}\sup_{0\leq u \leq s}\Big(\mathbb{E}|x_{\Delta}(u)|^p \Big)ds,
\end{align*}
where 
\begin{align*}
\nu_9&=|\xi(0)|^p+\nu_1T+\nu_6\\
\nu_{10}&=(\nu_1+\nu_8)\vee(\nu_1+\nu_7+\nu_8).
\end{align*}
As this holds for any $t\in [0,T]$, we then have
\begin{equation*}
  \sup_{0\leq u \leq t}(\mathbb{E}|x_{\Delta}(u)|^p)\leq \nu_9+2\nu_{10}\int_{0}^{t}\sup_{0\leq u \leq s}\Big(\mathbb{E}|x_{\Delta}(u)|^p \Big)ds.
\end{equation*}
The Gronwall inequality gives us
\begin{equation*}
  \sup_{0\leq u \leq T}(\mathbb{E}|x_{\Delta}(u)|^p)\leq c_3,
\end{equation*}
where $c_3=\nu_9e^{2\nu_{10}T}$ is independent of $\Delta$. The proof is thus complete.
\end{proof}
\subsection{Strong convergence}
\begin{lemma}\label{eq:L8}
Suppose Assumptions \ref{eq:a1}, \ref{eq:a2} and \ref{eq:a3}  hold and fix $T>0$. Then for any $\epsilon\in (0,1)$, there exists a pair $n=n(\epsilon)>0$ and $\bar{\Delta}=\bar{\Delta}(\epsilon)>0$ such that 
\begin{equation}\label{eq:27}
  \mathbb{P}(\varsigma_{\Delta,n}\leq T)\leq \epsilon
\end{equation}
as long as $\Delta \in (0, \bar \Delta]$, where
\begin{equation}\label{eq:28}
\varsigma_{\Delta,n}=\inf\{t\in [0,T]:x_{\Delta}(t)\notin (1/n,n)\}
\end{equation}
is a stopping time. 
\end{lemma}
\begin{proof}
Let $H(\cdot)$ be the Lyapunov function in \eqref{eq:13}. Then for $t\in [0,T]$, the It\^{o} formula gives us
\begin{align*}
\mathbb{E}(H(x_{\Delta}(t\wedge \varsigma_{\Delta,n}))-H(\xi(0)))&=\mathbb{E}\int_{0}^{t\wedge\varsigma_{\Delta,n}}\Big[H_x(x_{\Delta}(s^-))f_{\Delta}(\bar{x}_{\Delta}(s^-),\bar{r}_{\Delta}(s))\\
&+\frac{1}{2}H_{xx}(x_{\Delta}(s^-))\varphi(\bar{x}_{\Delta}((s-\tau)^-),\bar{r}_{\Delta}(s))^2g_{\Delta}(\bar{x}_{\Delta}(s^-))^2\\
&+\lambda\Big(H(x_{\Delta}(s^-)+h(\bar{x}_{\Delta}(s^-),\bar{r}_{\Delta}(s)))-H(x_{\Delta}(s^-))\Big)\Big]ds.
\end{align*}
For $s\in[0,t\wedge\varsigma_{\Delta,n}]$, we can expand to have 
\begin{align*}
&H_x(x_{\Delta}(s^-))f_{\Delta}(\bar{x}_{\Delta}(s^-),\bar{r}_{\Delta}(s))+\frac{1}{2}H_{xx}(x_{\Delta}(s^-))\varphi(\bar{x}_{\Delta}((s-\tau)^-),\bar{r}_{\Delta}(s))^2g_{\Delta}(\bar{x}_{\Delta}(s^-))^2\\
&+\lambda\Big(H(x_{\Delta}(s^-)+h(\bar{x}_{\Delta}(s^-),\bar{r}_{\Delta}(s)))-H(x_{\Delta}(s^-))\leq \mathcal{L}(x_{\Delta}(s^-),x_{\Delta}((s-\tau)^-),\bar{r}_{\Delta}(s))+\mathcal{J}_{21}+\mathcal{J}_{22}+\mathcal{J}_{23},
\end{align*}
where $\mathcal{L}H$ is the operator in \eqref{eq:5}, which now takes the form
\begin{align*}
\mathcal{L}(x_{\Delta}(s^-),x_{\Delta}((s-\tau)^-),\bar{r}_{\Delta}(s))
&= H_x(x_{\Delta}(s^-))f_{\Delta}(x_{\Delta}(s^-),\bar{r}_{\Delta}(s))\\
&+\frac{1}{2}H_{xx}\Big(x_{\Delta}(s^-))\varphi(x_{\Delta}((s-\tau)^-),\bar{r}_{\Delta}(s))^2g_{\Delta}(x_{\Delta}(s^-))^2\\
&+\lambda(H(x_{\Delta}(s^-)+h(x_{\Delta}(s^-),\bar{r}_{\Delta}(s))-H(x_{\Delta}(s^-)))\Big)
\end{align*}
with $H$ independent of $t$ and
\begin{align*}
\mathcal{J}_{21}&= H_x(x_{\Delta}(s^-))\Big(f_{\Delta}(\bar{x}_{\Delta}(s^-),\bar{r}_{\Delta}(s))-f_{\Delta}(x_{\Delta}(s^-),\bar{r}_{\Delta}(s))\Big)\\
\mathcal{J}_{22}&=\frac{1}{2} H_{xx}(x_{\Delta}(s^-))\Big(\varphi(\bar{x}_{\Delta}((s-\tau)^-),\bar{r}_{\Delta}(s))^2g_{\Delta}(\bar{x}_{\Delta}(s^-))^2-\varphi(x_{\Delta}((s-\tau)^-),\bar{r}_{\Delta}(s))^2g_{\Delta}(x_{\Delta}(s^-))^2\Big)\\
\mathcal{J}_{23}&=\lambda\Big(H(x_{\Delta}(s^-)+h(\bar{x}_{\Delta}(s^-),\bar{r}_{\Delta}(s)))-H(x_{\Delta}(s^-)+h(x_{\Delta}(s^-),\bar{r}_{\Delta}(s)))\Big).
\end{align*}
By Assumptions \ref{eq:a1} and \ref{eq:a2}, we can find a constant $\mathcal{K}_6$ such that
\begin{align}\label{eq:29}
\mathcal{L}(x_{\Delta}(s^-),x_{\Delta}((s-\tau)^-),\bar{r}_{\Delta}(s))\le \mathcal{K}_6.
\end{align}
Recalling from the definition of $f_{\Delta}$ and $g_{\Delta}$, we note for  $s\in[0,t\wedge\varsigma_{\Delta,n}]$ 
\begin{align*}
f_{\Delta}(x_{\Delta}(s^-),\bar{r}_{\Delta}(s))= f(x_{\Delta}(s^-),\bar{r}_{\Delta}(s)) \text{ and } g_{\Delta}(x_{\Delta}(s^-))=g(x_{\Delta}(s^-)).
\end{align*}
So for $s\in[0,t\wedge\varsigma_{\Delta,n}]$, we obtain from Lemma \ref{eq:L2} that
\begin{align*}
\mathcal{J}_{21}&\le K_nH_x(x_{\Delta}(s^-))\vert \bar{x}_{\Delta}(s^-)-x_{\Delta}(s^-)\vert.
\end{align*}
Moreover, for $s\in[0,t\wedge\varsigma_{\Delta,n}]$ and any $\bar{x}_{\Delta}(s^-),x_{\Delta}(s^-)\in [1/n,n]$, we note from \eqref{eq:17} that 
\begin{align*}
g(\bar{x}_{\Delta}(s^-))\vee g(x_{\Delta}(s^-))\le \mu(n).
\end{align*}
So for $s\in[0,t\wedge\varsigma_{\Delta,n}]$, we now obtain from Assumptions \ref{eq:a1} and \ref{eq:a3*}, and Lemma \ref{eq:L2}
\begin{align*}
\mathcal{J}_{22}&\le \frac{1}{2} H_{xx}(x_{\Delta}(s^-))\Big[g(x_{\Delta}(s^-))^2\Big( \varphi(\bar{x}_{\Delta}((s-\tau)^-),\bar{r}_{\Delta}(s))^2-\varphi(x_{\Delta}((s-\tau)^-),\bar{r}_{\Delta}(s))^2\Big)\\
&+ \varphi(x_{\Delta}((s-\tau)^-),\bar{r}_{\Delta}(s))^2\Big(g(\bar{x}_{\Delta}(s^-))^2- g(x_{\Delta}(s^-))^2\Big)\Big]\\
&\le H_{xx}(x_{\Delta}(s^-))\Big[L_n\sigma(\mu(n))^2\vert \bar{x}_{\Delta}(s-\tau)^--x_{\Delta}(s-\tau)^-\vert+K_n\sigma^2\mu(n))\vert \bar{x}_{\Delta}(s^-)-x_{\Delta}(s^-)\vert\Big].
\end{align*}
Also for $s\in[0,t\wedge\varsigma_{\Delta,n}]$, we obtain from the Lyapunov function in \eqref{eq:13} and the mean value theorem that
\begin{align*}
\mathcal{J}_{23}&\le \lambda\Big[\Big(x_{\Delta}(s^-)+h(\bar{x}_{\Delta}(s^-),\bar{r}_{\Delta}(s))\Big)^{{\phi}}-{\phi}\log\Big(x_{\Delta}(s^-)+h(\bar{x}_{\Delta}(s^-),\bar{r}_{\Delta}(s))\Big)\\
&-\Big(x_{\Delta}(s^-)+h(x_{\Delta}(s^-),\bar{r}_{\Delta}(s))\Big)^{{\phi}}+{\phi}\log\Big(x_{\Delta}(s^-)+h(x_{\Delta}(s^-),\bar{r}_{\Delta}(s))\Big)\Big]\\
&\le \lambda\Big[\Big(x_{\Delta}(s^-)+\alpha_3(i)\bar{x}_{\Delta}(s^-)\Big)^{{\phi}}-\Big(x_{\Delta}(s^-)+\alpha_3(i)x_{\Delta}(s^-)\Big)^{{\phi}}\\
&+{\phi}\log\Big(x_{\Delta}(s^-)+\alpha_3(i)x_{\Delta}(s^-)\Big)
-{\phi}\log\Big(x_{\Delta}(s^-)+\alpha_3(i)\bar{x}_{\Delta}(s^-)\Big)\Big]\\
&\le n\lambda\vert x_{\Delta}(s^-)+\alpha_3(i)\bar{x}_{\Delta}(s^-)-x_{\Delta}(s^-)-\alpha_3(i)x_{\Delta}(s^-)\vert\\
&+n\lambda\phi\vert x_{\Delta}(s^-)+\alpha_3(i)x_{\Delta}(s^-)-x_{\Delta}(s^-)-\alpha_3(i)\bar{x}_{\Delta}(s^-)\vert\\
&\le n\lambda\alpha_3(1+\phi)\vert \bar{x}_{\Delta}(s^-)-x_{\Delta}(s^-)\vert,
\end{align*}
where $\alpha_3=\max_{i\in \mathcal{S}}\alpha_3(i)$. Combining $\mathcal{J}_{21}$, $\mathcal{J}_{22}$ and $\mathcal{J}_{23}$ with \eqref{eq:29}, we now have
\begin{align*}
\mathbb{E}(H(x_{\Delta}(t\wedge \varsigma_{\Delta,n}))&\le H(\xi(0)))+\mathcal{K}_6T+K_n\mathbb{E}\int_{0}^{t\wedge\varsigma_{\Delta,n}} H_x(x_{\Delta}(s^-))\vert \bar{x}_{\Delta}(s^-)-x_{\Delta}(s^-)\vert ds\\
&+ \mathbb{E}\int_{0}^{t\wedge\varsigma_{\Delta,n}} H_{xx}(x_{\Delta}(s^-))\Big[L_n\sigma(\mu(n))^2\vert \bar{x}_{\Delta}(s-\tau)^--x_{\Delta}(s-\tau)^-\vert\\
&+K_n\sigma^2\mu(n)\vert \bar{x}_{\Delta}(s^-)-x_{\Delta}(s^-)\vert\Big]ds+n\lambda\alpha_3(1+\phi)\mathbb{E}\int_{0}^{t\wedge\varsigma_{\Delta,n}}\vert \bar{x}_{\Delta}(s^-)-x_{\Delta}(s^-)\vert ds\\
&\leq K7+ \mathcal{K}_8\mathbb{E}\int_{-\tau}^{0}|\xi([s/\Delta]\Delta)-\xi(s)|ds+(\mathcal{K}_8+\mathcal{K}_9)\int_{0}^T\mathbb{E}\Big(\mathbb{E}|\bar{x}_{\Delta}(s)-x_{\Delta}(s)|^p\Big|\mathcal{F}_{k(s)}\Big)^{1/p}ds
\end{align*}
where
\begin{align*}
\mathcal{K}_7=H(\xi(0))+ \mathcal{K}_6T \text{, }\mathcal{K}_8&=\max_{1/n\leq x \leq n}\{H_{xx}(x)\sigma (\mu(n))^2L_n\}
\end{align*}
and
\begin{align*}
\mathcal{K}_9&=\max_{1/n\leq x \leq n}\{H_x(x)K_n+H_{xx}(x)\sigma^2\mu(n) K_n+n\lambda\alpha_3(1+\phi)\}.
\end{align*}
So by Lemmas \ref{eq:L6} and \ref{eq:L6}, we now obtain
\begin{align*}
\mathbb{E}(H(x_{\Delta}(t\wedge \varsigma_{\Delta,n})))&\leq \mathcal{K}_7+ \mathcal{K}_3\mathcal{K}_8T\Delta^{\Upsilon}+(\mathcal{K}_8+\mathcal{K}_9)\mathfrak{C}_1^{1/p} \Big(\Delta^{p/2}(\psi(\Delta))^p+\Delta\Big)^{1/p}\\
&\times \int_{0}^{T}(\sup_{0\leq u \leq s}(\mathbb{E}|\bar{x}_{\Delta}(u)|^p))^{1/p}ds\\
&\leq \mathcal{K}_7+ \mathcal{K}_3\mathcal{K}_8T\Delta^{\Upsilon}+(\mathcal{K}_8+\mathcal{K}_9)\mathfrak{C}_1^{1/p} \Big(\Delta^{p/2}(\psi(\Delta))^p+\Delta\Big)^{1/p}c_3^{1/p}T.
\end{align*}
This implies
\begin{equation}\label{eq:30}
\mathbb{P}(\varsigma_{\Delta,n}\leq T)\leq \frac{ \mathcal{K}_7+ \mathcal{K}_3\mathcal{K}_8T\Delta^{\Upsilon}+(\mathcal{K}_8+\mathcal{K}_9)\mathfrak{C}_1^{1/p} \Big(\Delta^{p/2}(\psi(\Delta))^p+\Delta\Big)^{1/p}c_3^{1/p}T}{H(1/n)\wedge H(n)}.
\end{equation}
For any $\epsilon\in(0,1)$, we may select sufficiently large $n$ such that
\begin{equation}\label{eq:31}
\frac{\mathcal{K}_7}{H(1/n)\wedge H(n)}\leq \frac{\epsilon}{2}
\end{equation}
and sufficiently small of each step size $\Delta\in (0,\bar{\Delta}]$ such that
\begin{equation}\label{eq:32}
\frac{\mathcal{K}_3\mathcal{K}_8T\Delta^{\Upsilon}+(\mathcal{K}_8+\mathcal{K}_9)\mathfrak{C}_1^{1/p} \Big(\Delta^{p/2}(\psi(\Delta))^p+\Delta\Big)^{1/p}c_3^{1/p}T}{H(1/n)\wedge H(n)}\leq \frac{\epsilon}{2}.
\end{equation}
We can now combine \eqref{eq:31} and \eqref{eq:32} to obtain the required assertion.
\end{proof}
The following lemma shows the truncated EM scheme converges strongly in finite time.
\begin{lemma}\label{eq:L9}
Let Assumptions \ref{eq:a1}, \ref{eq:a3*}, \ref{eq:a2} and \ref{eq:a3} hold. Set
\begin{equation*}
 \vartheta_{\Delta,n}= \tau_n\wedge \varsigma_{\Delta,n},
\end{equation*}
where $\tau_n$ and $\varsigma_{\Delta,n}$ are \eqref{eq:12} and \eqref{eq:28}. Then for any $p\geq 2$, $T> 0$, we have for any sufficiently large $n$ and any $\Delta\in (0,\Delta^*]$, 
\begin{equation}\label{eq:33}
\mathbb{E}\Big( \sup_{0\leq t \leq T}|x_{\Delta}(t \wedge \vartheta_{\Delta,n})-x(t \wedge \vartheta_{\Delta,n})|^p  \Big)\leq \mathcal{C}\Big((\Delta+o(\Delta))(\psi(\Delta))^p) \vee \Delta^{p(1/4\wedge \Upsilon \wedge 1/p)}\Big)
\end{equation}
where $\mathcal{C}$ is a constant independent of $\Delta$ and consequently,
\begin{equation}\label{eq:34}
\lim_{\Delta\rightarrow 0}\mathbb{E}\Big( \sup_{0\leq t \leq T}|x_{\Delta}(t \wedge \vartheta_{\Delta,n})-x(t \wedge \vartheta_{\Delta,n})|^p \Big)=0.
\end{equation}
\end{lemma}
\begin{proof}
For $t_1\in [0,T]$, we obtain from \eqref{eq:4} and \eqref{eq:24} that
\begin{align}\label{eq:result2}
\mathbb{E}\Big( \sup_{0\leq t \leq t_1}|x_{\Delta}(t \wedge \vartheta_{\Delta,n})-x(t \wedge\vartheta_{\Delta,n})|^p\Big)&\le \mathcal{J}_{31}+\mathcal{J}_{32}+\mathcal{J}_{33}.
\end{align}
where
\begin{align*}
\mathcal{J}_{31}&=3^{p-1}\mathbb{E}\Big( |\int_{0}^{t_1\wedge \vartheta_{\Delta,n}}[f_{\Delta}(\bar{x}_{\Delta}(s^-),\bar{r}_{\Delta}(s))-f(x(s^-),r(s))]ds|^p\Big),\\
\mathcal{J}_{32}&=3^{p-1} \mathbb{E}\Big( \sup_{0\leq t \leq t_1}|\int_{0}^{t\wedge \vartheta_{\Delta,n}}[\varphi(\bar{x}_{\Delta}((s-\tau)^-),\bar{r}_{\Delta}(s))g_{\Delta}(\bar{x}_{\Delta}(s^-))\\
&-\varphi(x((s-\tau)^-),r(s))g(x(s^-))]dB(s)|^p\Big)\\
\mathcal{J}_{33}&= 3^{p-1}\mathbb{E}\Big( \sup_{0\leq t \leq t_1}|\int_{0}^{t\wedge \vartheta_{\Delta,n}}[h(\bar{x}_{\Delta}(s^-),\bar{r}_{\Delta}(s))-h(x(s^-),r(s))]dN(s)|^p\Big).
\end{align*}
By the  H\"older and elementary inequalities, we compute
\begin{align*}
\mathcal{J}_{31}&\le \mathcal{J}_{311}+\mathcal{J}_{312},
\end{align*}
where
\begin{align*}
 \mathcal{J}_{311}&=3^{p-1}2^{p-1}T^{p-1}\mathbb{E} \int_{0}^{t_1\wedge \vartheta_{\Delta,n}}|f_{\Delta}(\bar{x}_{\Delta}(s^-),r(s))-f(x(s^-),r(s))|^pds\\
\mathcal{J}_{312}&=3^{p-1}2^{p-1}T^{p-1}\mathbb{E}\int_{0}^{t_1\wedge \vartheta_{\Delta,n}}|f_{\Delta}(\bar{x}_{\Delta}(s^-),\bar{r}_{\Delta}(s))-f_{\Delta}(\bar{x}(s^-),r(s))|^pds.
\end{align*}
It is clear from the definition of the truncated function $f_{\Delta}$ that $f_{\Delta}(\bar{x}_{\Delta}(s^-),\bar{r}_{\Delta}(s))=f(\bar{x}_{\Delta}(s^-),\bar{r}_{\Delta}(s))$ for $s\in [0,t_1\wedge \vartheta_{\Delta,n}]$ So by Lemma \ref{eq:L2},
\begin{align*}
\mathcal{J}_{311}&=3^{p-1}2^{p-1}T^{p-1}\mathbb{E} \int_{0}^{t_1\wedge \vartheta_{\Delta,n}}|f(\bar{x}_{\Delta}(s^-),r(s))-f(x(s^-),r(s))|^pds\\
&\le 3^{p-1}2^{p-1}T^{p-1}K_n^p\mathbb{E}\int_{0}^{t_1\wedge \vartheta_{\Delta,n}}|\bar{x}_{\Delta}(s^-)-x(s^-)|^pds.
\end{align*}
Let $n=[T/\Delta]$ be the integer part of $T/\Delta$. Then 
\begin{align*}
\mathcal{J}_{312}&=3^{p-1}2^{p-1}T^{p-1}\mathbb{E}\int_{0}^{T}|f_{\Delta}(\bar{x}_{\Delta}(s^-),\bar{r}_{\Delta}(s))-f_{\Delta}(\bar{x}(s^-),r(s))|^pds\\
&=3^{p-1}2^{p-1}T^{p-1}\sum_{k=0}^{n}\mathbb{E}\int_{t_k}^{t_{k+1}}|f_{\Delta}(\bar{x}_{\Delta}(t_k),r(t_k))-f_{\Delta}(\bar{x}(t_k),r(s))|^pds,
\end{align*}
with $t_{n+1}$ now set to be $T$. We now have from \eqref{eq:19}
\begin{align}\label{markov1}
\mathcal{J}_{312}&\le 3^{p-1}2^{2(p-1)}T^{p-1}\sum_{k=0}^{n}\mathbb{E}\int_{t_k}^{t_{k+1}}[|f_{\Delta}(\bar{x}_{\Delta}(t_k),r(t_k))|^p+|f_{\Delta}(\bar{x}(t_k),r(s))|^p]1_{\{r(s)\neq r(t_k)\}}ds\nonumber\\
&\le 3^{p-1}2^{2(p-1)}T^{p-1}\sum_{k=0}^{n}\int_{t_k}^{t_{k+1}}\mathbb{E}\Big[\mathbb{E}[((\psi(\Delta))^p+(\psi(\Delta))^p)1_{\{r(s)\neq r(t_k)\}}|r(t_k)]\Big]ds\nonumber\\
&= 3^{p-1}2^{2(p-1)}T^{p-1}\sum_{k=0}^{n}\int_{t_k}^{t_{k+1}}\mathbb{E}\Big[\mathbb{E}[2(\psi(\Delta))^p|r(t_k)]\mathbb{E}[1_{\{r(s)\neq r(t_k)\}}|r(t_k)]\Big]ds,
\end{align}
where we use the fact that $\bar{x}(s)$ and $1_{r(s)\neq r(t_k)}$ are conditionally independent with respect to the $\sigma-$algebra generated by $r(t_k)$ in the last step. By the Markov property, we compute
\begin{align}\label{markov2}
\mathbb{E}[1_{\{r(s)\neq r(t_k)\}}|r(t_k)]&=\sum_{i\in \mathcal{S}}1_{\{r(t_k)=i\}}\mathbb{P}(r(s)\neq i|r(t_k)=i)\nonumber\\
&=\sum_{i\in \mathcal{S}}1_{\{r(t_k)=i\}}\sum_{i\neq j}(\gamma_{ij}(s-t_k)+o(s-t_k))\nonumber\\
&\le (\max_{1\le i \le N}(-\gamma_{ij})\Delta+o(\Delta))\sum_{i\in \mathcal{S}}1_{\{r(t_k)=i\}}\nonumber\\
&\le \bar{c}_1\Delta+o(\Delta).
\end{align}
where $\bar{c}_1=\max_{1\le i \le N}(-\gamma_{ij})$. By Lemma \ref{eq:L7}, we note
\begin{align*}
\mathbb{E}\int_{t_k}^{t_{k+1}}|f_{\Delta}(\bar{x}_{\Delta}(t_k),r(t_k))-f(\bar{x}(t_k),r(s))|^pds &\le 2(\bar{c}_1\Delta+o(\Delta)) \int_{t_k}^{t_{k+1}}(\psi(\Delta))^pds\\
&\le 2(\bar{c}_1\Delta+o(\Delta))\Delta(\psi(\Delta))^p.
\end{align*}
This implies
\begin{align*}
\mathbb{E}\int_{0}^{T}|f_{\Delta}(\bar{x}_{\Delta}(s^-),\bar{r}_{\Delta}(s))-f(\bar{x}(s^-),r(s))|^pds&\le 2(\bar{c}_1\Delta+o(\Delta))(\psi(\Delta))^p
\end{align*}
and consequently 
\begin{align*}
\mathbb{E}\int_{0}^{t_1\wedge \vartheta_{\Delta,n}}|f_{\Delta}(\bar{x}_{\Delta}(s^-),\bar{r}_{\Delta}(s))-f(\bar{x}(s^-),r(s))|^pds&\le 2(\bar{c}_1\Delta+o(\Delta))(\psi(\Delta))^p
\end{align*}
Substituting this into $\mathcal{J}_{312}$ yields 
\begin{align*}
\mathcal{J}_{312}&\le 3^{p-1}2^{2p-1}T^{p-1}(\bar{c}_1\Delta+o(\Delta))(\psi(\Delta))^p.
\end{align*}
We then combine $\mathcal{J}_{311}$ and $\mathcal{J}_{312}$ to obtain 
\begin{align*}
\mathcal{J}_{31}&\le \bar{c}_2(\bar{c}_1\Delta+o(\Delta))(\psi(\Delta))^p+\bar{c}_3\mathbb{E}\int_{0}^{t_1\wedge \vartheta_{\Delta,n}}|\bar{x}_{\Delta}(s^-)-x(s^-)|^pds,
\end{align*}
where 
\begin{align*}
\bar{c}_2&=3^{p-1}2^{2p-1}T^{p-1}\\
\bar{c}_3&=3^{p-1}2^{p-1}T^{p-1}K_n^p.
\end{align*} 
Also by the H\"older and Burkholder-Davis Gundy inequalities, we have
\begin{align*}
\mathcal{J}_{32}&\leq 3^{p-1}T^{\frac{p-2}{2}}C_1(p) \mathbb{E}\int_{0}^{t_1\wedge \vartheta_{\Delta,n}}\Big(|\varphi(\bar{x}_{\Delta}((s-\tau)^-),\bar{r}_{\Delta}(s))g_{\Delta}(\bar{x}_{\Delta}(s^-))\\
&-\varphi(x((s-\tau)^-),r(s))g_{\Delta}(\bar{x}_{\Delta}(s^-))+\varphi(x((s-\tau)^-),r(s))g_{\Delta}(\bar{x}_{\Delta}(s^-))\\
&-\varphi(x((s-\tau)^-),r(s))g(x(s^-))|^p\Big)ds\\
&\le \mathcal{J}_{321}+\mathcal{J}_{322},
\end{align*}
where 
\begin{align*}
\mathcal{J}_{321}&= 2^{p-1}3^{p-1}T^{\frac{p-2}{2}}C_1(p) \mathbb{E}\int_{0}^{t_1\wedge  \vartheta_{\Delta,n}}g_{\Delta}(\bar{x}_{\Delta}(s^-))^p|\varphi(\bar{x}_{\Delta}((s-\tau)^-),\bar{r}_{\Delta}(s))-\varphi(x((s-\tau)^-),r(s))|^pds\\
\mathcal{J}_{322}&=2^{p-1}3^{p-1}T^{\frac{p-2}{2}}C_1(p) \mathbb{E}\int_{0}^{t_1\wedge \vartheta_{\Delta,n}}\varphi(x((s-\tau)^-),r(s))^p|g_{\Delta}(\bar{x}_{\Delta}(s^-))-g(x(s^-))|^p ds.
\end{align*}
where $C_1(p)$ is a positive constant. For $s\in[0,t_1\wedge \vartheta_{\Delta,n}]$, we note from \eqref{eq:17} that $\bar{x}_{\Delta}(s^-)\in [1/n,n]$ and $g_{\Delta}(\bar{x}_{\Delta}(s^-))\le \mu(n)$. So we now have
\begin{align*}
\mathcal{J}_{321}&\le 2^{p-1}3^{p-1}T^{\frac{p-2}{2}}C_1(p) (\mu(n))^p\mathbb{E}\int_{0}^{t_1\wedge \vartheta_{\Delta,n}}|\varphi(\bar{x}_{\Delta}((s-\tau)^-),\bar{r}_{\Delta}(s))-\varphi(x((s-\tau)^-),r(s))|^pds\\
&\le \mathcal{J}_{323}+\mathcal{J}_{324},
\end{align*}
where
\begin{align*}
\mathcal{J}_{323}&=2^{2(p-1)}3^{p-1}T^{\frac{p-2}{2}}C_1(p) (\mu(n))^p\mathbb{E}\int_{0}^{t_1\wedge \vartheta_{\Delta,n}}|\varphi(\bar{x}_{\Delta}((s-\tau)^-),r(s))-\varphi(x((s-\tau)^-),r(s))|^pds\\
\mathcal{J}_{324}&=2^{2(p-1)}3^{p-1}T^{\frac{p-2}{2}}C_1(p) (\mu(n))^p\mathbb{E}\int_{0}^{t_1\wedge \vartheta_{\Delta,n}}|\varphi(\bar{x}_{\Delta}((s-\tau)^-),\bar{r}_{\Delta}(s))-\varphi(\bar{x}((s-\tau)^-),r(s))|^pds
\end{align*}
By Assumption \ref{eq:a3*}, we obtain
\begin{align*}
\mathcal{J}_{323}&\le 2^{2(p-1)}3^{p-1}T^{\frac{p-2}{2}}C_1(p) (\mu(n))^pL_n^p\mathbb{E}\int_{0}^{t_1\wedge \vartheta_{\Delta,n}}|\bar{x}_{\Delta}((s-\tau)^-)-x((s-\tau)^-)|^pds.
\end{align*}
Also as before, we compute
\begin{align*}
\mathcal{J}_{324}&=2^{2(p-1)}3^{p-1}T^{\frac{p-2}{2}}C_1(p) (\mu(n))^p\mathbb{E}\int_{0}^{T}|\varphi(\bar{x}_{\Delta}((s-\tau)^-),\bar{r}_{\Delta}(s))-\varphi(\bar{x}((s-\tau)^-),r(s))|^pds\\
&=2^{p-1}3^{p-1}T^{\frac{p-2}{2}}C_1(p) (\mu(n))^p\sum_{k=0}^{n}\mathbb{E}\int_{t_k}^{t_{k+1}}|\varphi(\bar{x}_{\Delta}((t_k-\tau)^-),\bar{r}_{\Delta}(t_k))-\varphi(\bar{x}((t_k-\tau)^-),r(s))|^pds,
\end{align*}
where $n$ is the usual integer part of $T/\Delta$ with $t_{n+1}$ set to be $T$. By elementary inequality, 
\begin{align*}
\mathcal{J}_{324}&\le 2^{2(p-1)}3^{p-1}T^{\frac{p-2}{2}}C_1(p) (\mu(n))^p\sum_{k=0}^{n}\mathbb{E}\int_{t_k}^{t_{k+1}}[|\varphi(\bar{x}_{\Delta}((t_k-\tau)^-),\bar{r}_{\Delta}(t_k))|^p\\
&+|\varphi(\bar{x}((t_k-\tau)^-),r(s))|^p1_{\{r(s)\neq r(t_k)\}}]ds\\
&= 2^{2(p-1)}3^{p-1}T^{\frac{p-2}{2}}C_1(p) (\mu(n))^p\sum_{k=0}^{n}\int_{t_k}^{t_{k+1}}\mathbb{E}\Big[\mathbb{E}[|\varphi(\bar{x}_{\Delta}((t_k-\tau)^-),\bar{r}_{\Delta}(t_k))|^p\\
&+|\varphi(\bar{x}((t_k-\tau)^-),r(s))|^p1_{\{r(s)\neq r(t_k)\}}|r(t_k)]\Big]ds\\
&= 2^{2(p-1)}3^{p-1}T^{\frac{p-2}{2}}C_1(p) (\mu(n))^p\sum_{k=0}^{n}\int_{t_k}^{t_{k+1}}\mathbb{E}\Big[\mathbb{E}[|\varphi(\bar{x}_{\Delta}((t_k-\tau)^-),\bar{r}_{\Delta}(t_k))|^p\\
&+|\varphi(\bar{x}((t_k-\tau)^-),r(s))|^p]\mathbb{E}[1_{\{r(s)\neq r(t_k)\}}|r(t_k)]\Big]ds
\end{align*}
We note from \eqref{markov2} that
\begin{align*}
\mathbb{E}[1_{\{r(s)\neq r(t_k)\}}|r(t_k)]&\le \bar{c}_1\Delta+o(\Delta).
\end{align*}
By Assumption \ref{eq:a1}, we have
\begin{align*}
\mathbb{E}\int_{t_k}^{t_{k+1}}|\varphi(\bar{x}_{\Delta}((t_k-\tau)^-),\bar{r}_{\Delta}(t_k))-\varphi(\bar{x}((t_k-\tau)^-),r(s))|^pds
&\le 2(\bar{c}_1\Delta+o(\Delta))\int_{t_k}^{t_{k+1}}\sigma^pds\\
&\le 2\sigma^p(\bar{c}_1\Delta+o(\Delta))\Delta.
\end{align*}
This means by Assumption \ref{eq:a1}, we have
\begin{align*}
\mathbb{E}\int_{0}^{T}|\varphi(\bar{x}_{\Delta}((s-\tau)^-),\bar{r}_{\Delta}(s))-\varphi(\bar{x}((s-\tau)^-),r(s))|^pds\le 2\sigma^p(\bar{c}_1\Delta+o(\Delta))
\end{align*}
and hence,
\begin{align*}
\mathbb{E}\int_{0}^{t_1\wedge \vartheta_{\Delta,n}}|\varphi(\bar{x}_{\Delta}((s-\tau)^-),\bar{r}_{\Delta}(s))-\varphi(\bar{x}((s-\tau)^-),r(s))|^pds\le 2\sigma^p(\bar{c}_1\Delta+o(\Delta)).
\end{align*}
Inserting this into $\mathcal{J}_{324}$ yields
\begin{align*}
\mathcal{J}_{324}\le 2^{2p-1}3^{p-1}T^{\frac{p-2}{2}}C_1(p) (\mu(n))^p\sigma^p(\bar{c}_1\Delta+o(\Delta)).
\end{align*}
We obtain from $\mathcal{J}_{323}$ and $\mathcal{J}_{324}$
\begin{align*}
\mathcal{J}_{321}&\le 2^{2p-1}3^{p-1}T^{\frac{p-2}{2}}C_1(p) (\mu(n))^p\sigma^p(\bar{c}_1\Delta+o(\Delta))\\
&+2^{2(p-1)}3^{p-1}T^{\frac{p-2}{2}}C_1(p) (\mu(n))^pL_n^p\mathbb{E}\int_{0}^{t_1\wedge \vartheta_{\Delta,n}}|\bar{x}_{\Delta}((s-\tau)^-)-x((s-\tau)^-)|^pds.
\end{align*}
Moreover, by Assumption \ref{eq:a1} and Lemma \ref{eq:L2}
\begin{align*}
\mathcal{J}_{322}&\le 2^{p-1}3^{p-1}T^{\frac{p-2}{2}}C_1(p)(\mu(n))^pK_n^p \mathbb{E}\int_{0}^{t_1\wedge \vartheta_{\Delta,n}}|\bar{x}_{\Delta}(s^-)-x(s^-)|^p ds.
\end{align*}
Combining $\mathcal{J}_{321}$ and $\mathcal{J}_{322}$, we have 
\begin{align*}
\mathcal{J}_{32}&\le \bar{c}_4(\bar{c}_1\Delta+o(\Delta))+\bar{c}_5\mathbb{E}\int_{0}^{t_1\wedge \vartheta_{\Delta,n}}|\bar{x}_{\Delta}((s-\tau)^-)-x((s-\tau)^-)|^pds\\
&+\bar{c}_6\mathbb{E}\int_{0}^{t_1\wedge \vartheta_{\Delta,n}}|\bar{x}_{\Delta}(s^-)-x(s^-)|^p ds,
\end{align*}
where
\begin{align*}
\bar{c}_4&=2^{2p-1}3^{p-1}T^{\frac{p-2}{2}}C_1(p) (\mu(n))^p\sigma^p\\
\bar{c}_5&=2^{2(p-1)}3^{p-1}T^{\frac{p-2}{2}}C_1(p) (\mu(n))^pL_n^p\\
\bar{c}_6&=2^{p-1}3^{p-1}T^{\frac{p-2}{2}}C_1(p)(\mu(n))^pK_n^p.
\end{align*}
Furthermore, by elementary inequality
\begin{align*}
\mathcal{J}_{33}&=3^{p-1}\mathbb{E}\Big( \sup_{0\leq t \leq t_1}|\int_{0}^{t\wedge \vartheta_{\Delta,n}}[h(\bar{x}_{\Delta}(s^-),\bar{r}_{\Delta}(s))-h(x(s^-),r(s))]d\widetilde{N}(s)\\
&+ \lambda\int_{0}^{t \wedge \vartheta_{\Delta,n}}[h(\bar{x}_{\Delta}(s^-),\bar{r}_{\Delta}(s))-h(x(s^-),r(s))]ds|^p\Big)\\
&\le \mathcal{J}_{331}+\mathcal{J}_{332},
\end{align*}
where
\begin{align*}
\mathcal{J}_{331}&=2^{p-1}3^{p-1}\mathbb{E}\Big( \sup_{0\leq t \leq t_1}|\int_{0}^{t\wedge \vartheta_{\Delta,n}}[h(\bar{x}_{\Delta}(s^-),\bar{r}_{\Delta}(s))-h(x(s^-),r(s))]d\widetilde{N}(s)|^p\Big)\\
\mathcal{J}_{332}&=2^{p-1}3^{p-1}\lambda^p\mathbb{E}\Big(\sup_{0\leq t \leq t_1}|\int_{0}^{t \wedge \vartheta_{\Delta,n}}[h(\bar{x}_{\Delta}(s^-),\bar{r}_{\Delta}(s))-h(x(s^-),r(s))]ds|^p\Big).
\end{align*}
By the Doob martingale inequality and martingale isometry, we have
\begin{align*}
\mathcal{J}_{331}&\le 2^{p-1}3^{p-1}\lambda^{\frac{p}{2}}C_2(p)\mathbb{E}\Big( \sup_{0\leq t \leq t_1}|\int_{0}^{t\wedge \vartheta_{\Delta,n}}|h(\bar{x}_{\Delta}(s^-),\bar{r}_{\Delta}(s))-h(x(s^-),r(s))|^2d\widetilde{N}(s)\Big)^{\frac{p}{2}}\\
&\le 2^{p-1}3^{p-1}\lambda^{p/2}T^{\frac{p-2}{2}}C_2(p)\mathbb{E}\int_{0}^{t\wedge \vartheta_{\Delta,n}}|h(\bar{x}_{\Delta}(s^-),\bar{r}_{\Delta}(s))-h(x(s^-),r(s))|^pds\\
&\le \mathcal{J}_{333}+\mathcal{J}_{334},
\end{align*}
where
\begin{align*}
\mathcal{J}_{333}&= 2^{p-1}3^{p-1}\lambda^{p/2}T^{\frac{p-2}{2}}C_2(p)\mathbb{E}\int_{0}^{t\wedge \vartheta_{\Delta,n}}|h(\bar{x}_{\Delta}(s^-),r(s))-h(x(s^-),r(s))|^pds\\
\mathcal{J}_{334}&= 2^{p-1}3^{p-1}\lambda^{p/2}T^{\frac{p-2}{2}}C_2(p)\mathbb{E}\int_{0}^{t\wedge \vartheta_{\Delta,n}}|h(\bar{x}_{\Delta}(s^-),\bar{r}_{\Delta}(s))-h(\bar{x}_{\Delta}(s^-),r(s))|^pds
\end{align*}
and $C_2(p)$ is a positive constant. By Lemma \ref{eq:L2},
\begin{align*}
\mathcal{J}_{333}&\le 2^{p-1}3^{p-1}\lambda^{p/2}T^{\frac{p-2}{2}}C_2(p)K_n^p\mathbb{E}\int_{0}^{t\wedge \vartheta_{\Delta,n}}|\bar{x}_{\Delta}(s^-)-x(s^-)|^pds.
\end{align*}
We also compute
\begin{align*}
\mathcal{J}_{334}&=2^{p-1}3^{p-1}\lambda^{p/2}T^{\frac{p-2}{2}}C_2(p)\mathbb{E}\int_{0}^{T}|h(\bar{x}_{\Delta}(s^-),\bar{r}_{\Delta}(s))-h(\bar{x}_{\Delta}(s^-),r(s))|^pds\\
&= 2^{p-1}3^{p-1}\lambda^{p/2}T^{\frac{p-2}{2}}C_2(p)\sum_{k=0}^{n}\mathbb{E}\int_{t_k}^{t_{k+1}}|h(\bar{x}_{\Delta}(t_k),\bar{r}_{\Delta}(t_k))-h(\bar{x}_{\Delta}(t_k),r(s))|^pds\\
&\le 2^{2(p-1)}3^{p-1}\lambda^{p/2}T^{\frac{p-2}{2}}C_2(p)\sum_{k=0}^{n}\mathbb{E}\int_{t_k}^{t_{k+1}}[|h(\bar{x}_{\Delta}(t_k),\bar{r}_{\Delta}(t_k))|^p+|h(\bar{x}_{\Delta}(t_k),r(s))|^p1_{\{r(s)\neq r(t_k)\}}ds\\
&\le 2^{2(p-1)}3^{p-1}\lambda^{p/2}T^{\frac{p-2}{2}}C_2(p)\sum_{k=0}^{n}\int_{t_k}^{t_{k+1}}\mathbb{E}\Big[\mathbb{E}[|h(\bar{x}_{\Delta}(t_k),\bar{r}_{\Delta}(t_k))|^p\\
&+|h(\bar{x}_{\Delta}(t_k),r(s))|^p|r(t_k))] \mathbb{E}[1_{\{r(s)\neq r(t_k)\}}|r(t_k)]\Big]ds,
\end{align*}
where $n$, as usual, is the integer part of $T/\Delta$ with $t_{n+1}$ set to be $T$. By Lemma \ref{eq:L7} and \eqref{markov2}
\begin{align*}
\mathbb{E}\int_{t_k}^{t_{k+1}}|h(\bar{x}_{\Delta}(t_k),\bar{r}_{\Delta}(t_k))-h(\bar{x}_{\Delta}(t_k),r(s))|^pds &\le (\bar{c}_1\Delta+o(\Delta))\int_{t_k}^{t_{k+1}}2\alpha_3(i)\mathbb{E}|\bar{x}_{\Delta}(t_k)|^p ds\\
&\le 2\alpha_3(\bar{c}_1\Delta+o(\Delta))\Delta,
\end{align*}
where $\alpha_3=\max_{i\in \mathcal{S}}\alpha_3(i)$. Consequently, we have
\begin{align*}
\mathbb{E}\int_{0}^{T}|h(\bar{x}_{\Delta}(s^-),\bar{r}_{\Delta}(s))-h(\bar{x}_{\Delta}(s^-),r(s))|^pds
&\le 2\alpha_3(\bar{c}_1\Delta+o(\Delta))
\end{align*}
and then,
\begin{align}\label{eq:pois}
\mathbb{E}\int_{0}^{t \wedge \vartheta_{\Delta,n}}|h(\bar{x}_{\Delta}(s^-),\bar{r}_{\Delta}(s))-h(\bar{x}_{\Delta}(s^-),r(s))|^pds
&\le 2\alpha_3(\bar{c}_1\Delta+o(\Delta)).
\end{align}
We substitute this into $\mathcal{J}_{334}$ to get
\begin{align*}
\mathcal{J}_{334}&\le 2^{2p-1}3^{p-1}C_2(p)\alpha_3\lambda^{p/2}T^{\frac{p-2}{2}}(\bar{c}_1\Delta+o(\Delta)).
\end{align*}
It then follows from $\mathcal{J}_{333}$ and $\mathcal{J}_{334}$ that
\begin{align*}
\mathcal{J}_{331}&\le  2^{2p-1}3^{p-1}C_2(p)\alpha_3\lambda^{p/2}T^{\frac{p-2}{2}}(\bar{c}_1\Delta+o(\Delta))\\
&+ 2^{p-1}3^{p-1}C_2(p)K_n^p\lambda^{p/2}T^{\frac{p-2}{2}}\mathbb{E}\int_{0}^{t\wedge \vartheta_{\Delta,n}}|\bar{x}_{\Delta}(s^-)-x(s^-)|^pds.
\end{align*}
By the the H\"older inequality,
\begin{align*}
\mathcal{J}_{332}&\le 2^{p-1}3^{p-1}\lambda^pT^{p-1}\mathbb{E}\int_{0}^{t \wedge \vartheta_{\Delta,n}}|h(\bar{x}_{\Delta}(s^-),\bar{r}_{\Delta}(s))-h(x(s^-),r(s))|^pds\\
&=\mathcal{J}_{335}+\mathcal{J}_{336},
\end{align*}
where
\begin{align*}
\mathcal{J}_{335}&= 2^{p-1}3^{p-1}\lambda^pT^{p-1}\mathbb{E}\int_{0}^{t\wedge \vartheta_{\Delta,n}}|h(\bar{x}_{\Delta}(s^-),r(s))-h(x(s^-),r(s))|^pds\\
\mathcal{J}_{336}&=  2^{p-1}3^{p-1}\lambda^pT^{p-1}\mathbb{E}\int_{0}^{t\wedge \vartheta_{\Delta,n}}|h(\bar{x}_{\Delta}(s^-),\bar{r}_{\Delta}(s))-h(\bar{x}_{\Delta}(s^-),r(s))|^pds.
\end{align*}
So by Lemma \ref{eq:L2},
\begin{align*}
\mathcal{J}_{335}&= 2^{p-1}3^{p-1}\lambda^pT^{p-1}K_n^p\mathbb{E}\int_{0}^{t\wedge \vartheta_{\Delta,n}}|\bar{x}_{\Delta}(s^-)-x(s^-)|^pds.
\end{align*}
Apparently, we see from \eqref{eq:pois} that
 \begin{align*}
\mathbb{E}\int_{0}^{t \wedge \vartheta_{\Delta,n}}|h(\bar{x}_{\Delta}(s^-),\bar{r}_{\Delta}(s))-h(\bar{x}_{\Delta}(s^-),r(s))|^pds
&\le 2\alpha_3(\bar{c}_1\Delta+o(\Delta)). 
\end{align*}
This implies 
\begin{align*}
\mathcal{J}_{336}&\le 2^{p-1}3^{p-1}2\alpha_3\lambda^pT^{p-1}(\bar{c}_1\Delta+o(\Delta)). 
\end{align*}
We now have from $\mathcal{J}_{335}$ and $\mathcal{J}_{336}$
\begin{align*}
\mathcal{J}_{332}&\le 2^{p-1}3^{p-1}2\alpha_3\lambda^pT^{p-1}(\bar{c}_1\Delta+o(\Delta))\\
&+2^{p-1}3^{p-1}\lambda^pT^{p-1}K_n^p\mathbb{E}\int_{0}^{t\wedge \vartheta_{\Delta,n}}|\bar{x}_{\Delta}(s^-)-x(s^-)|^pds.
\end{align*}
We then combine $\mathcal{J}_{331}$ and $\mathcal{J}_{332}$ to have
\begin{align*}
\mathcal{J}_{33}&\le  \bar{c}_7(\bar{c}_1\Delta+o(\Delta))+\bar{c}_8\mathbb{E}\int_{0}^{t\wedge \vartheta_{\Delta,n}}|\bar{x}_{\Delta}(s^-)-x(s^-)|^pds+\bar{c}_9\mathbb{E}\int_{0}^{t\wedge \vartheta_{\Delta,n}}|\bar{x}_{\Delta}(s^-)-x(s^-)|^pds,
\end{align*}
where
\begin{align*}
\bar{c}_7&=2^{2p-1}3^{p-1}C_2(p)\alpha_3\lambda^{p/2}T^{\frac{p-2}{2}}+2^{p-1}3^{p-1}2\alpha_3\lambda^pT^{p-1}\\
\bar{c}_8&=2^{p-1}3^{p-1}C_2(p)K_n^p\lambda^{p/2}T^{\frac{p-2}{2}}\\
\bar{c}_9&=2^{p-1}3^{p-1}\lambda^pT^{p-1}K_n^p.
\end{align*}
Substituting $\mathcal{J}_{31}$, $\mathcal{J}_{32}$ and $\mathcal{J}_{33}$ into \eqref{eq:result2}, we get
\begin{align*}
&\mathbb{E}\Big( \sup_{0\leq t \leq t_1}|x_{\Delta}(t \wedge \vartheta_{\Delta,n})-x(t \wedge\vartheta_{\Delta,n})|^p\Big)\\
&\le \bar{c}_2(\bar{c}_1\Delta+o(\Delta))(\psi(\Delta))^p+\bar{c}_4(\bar{c}_1\Delta+o(\Delta))+\bar{c}_7(\bar{c}_1\Delta+o(\Delta))\\
&+\bar{c}_5\mathbb{E}\int_{0}^{t_1\wedge \vartheta_{\Delta,n}}|\bar{x}_{\Delta}((s-\tau)^-)-x((s-\tau)^-)|^pds+\bar{c}_6\mathbb{E}\int_{0}^{t_1\wedge \vartheta_{\Delta,n}}|\bar{x}_{\Delta}(s^-)-x(s^-)|^p ds\\
&+\bar{c}_8\mathbb{E}\int_{0}^{t\wedge \vartheta_{\Delta,n}}|\bar{x}_{\Delta}(s^-)-x(s^-)|^pds+\bar{c}_9\mathbb{E}\int_{0}^{t\wedge \vartheta_{\Delta,n}}|\bar{x}_{\Delta}(s^-)-x(s^-)|^pds.
\end{align*}
It then follows that
\begin{align*}
&\mathbb{E}\Big( \sup_{0\leq t \leq t_1}|x_{\Delta}(t \wedge \vartheta_{\Delta,n})-x(t \wedge\vartheta_{\Delta,n})|^p\Big)\\
&\le \bar{c}_{10}(\bar{c}_1\Delta+o(\Delta))(\psi(\Delta))^p+\bar{c}_5\mathbb{E}\int_{-\tau}^{0}|\xi([s/\Delta]\Delta)-\xi(s)|^pds+\bar{c}_{11}\mathbb{E}\int_{0}^{t\wedge \vartheta_{\Delta,n}}|\bar{x}_{\Delta}(s^-)-x(s^-)|^pds,
\end{align*}
where 
\begin{align*}
\bar{c}_{10}&=\bar{c}_2\vee \bar{c}_4\vee \bar{c}_7\\
\bar{c}_{11}&=\bar{c}_5\vee \bar{c}_6\vee \bar{c}_8\vee \bar{c}_9.
\end{align*}
By elementary inequality, Assumption \ref{eq:a3} and Lemma \ref{eq:L6}
\begin{align*}
&\mathbb{E}\Big( \sup_{0\leq t \leq t_1}|x_{\Delta}(t \wedge \vartheta_{\Delta,n})-x(t \wedge\vartheta_{\Delta,n})|^p\Big)\\
&\le \bar{c}_{10}(\bar{c}_1\Delta+o(\Delta))(\psi(\Delta))^p)+\bar{c}_5\Delta^{p\Upsilon}\tau+2^{p-1}\bar{c}_{11}\int_{0}^{T}\mathbb{E}\Big(\mathbb{E}|\bar{x}_{\Delta}(s)-x_{\Delta}(s)|^p\big|\mathcal{F}_{k(s)}\Big)ds\\
&+2^{p-1}\bar{c}_{11}\int_{0}^{t_1}\mathbb{E}\Big(\sup_{0\leq t\leq s}|x_{\Delta}(t\wedge \vartheta{\Delta,n})-x(t\wedge \vartheta{\Delta,n})|^p\Big)ds\\
&\le \bar{c}_{10}(\bar{c}_1\Delta+o(\Delta))(\psi(\Delta))^p)+\bar{c}_5\Delta^{p\Upsilon}\tau+2^{p-1}\bar{c}_{11}\mathfrak{C}_1\Big(\Delta^{p/2}(\psi(\Delta))^p+\Delta\Big)\int_{0}^{T}\mathbb{E}|\bar{x}_{\Delta}(s)|^pds\\
&+2^{p-1}\bar{c}_{11}\int_{0}^{t_1}\mathbb{E}\Big(\sup_{0\leq t\leq s}|x_{\Delta}(t\wedge \vartheta_{\Delta,n})-x(t\wedge \vartheta_{\Delta,n})|^p\Big)ds
\end{align*}
So by Lemma \ref{eq:L7} and noting that $(\Delta^{1/4}(\psi(\Delta)))^{p}\le 1$, we now have
\begin{align*}
&\mathbb{E}\Big( \sup_{0\leq t \leq t_1}|x_{\Delta}(t \wedge \vartheta_{\Delta,n})-x(t \wedge\vartheta_{\Delta,n})|^p\Big)\\
&\le \bar{c}_{10}(\bar{c}_1\Delta+o(\Delta))(\psi(\Delta))^p)+\bar{c}_5\Delta^{p\Upsilon}\tau+2^{p-1}\bar{c}_{11}c_3\mathfrak{C}_1\Big([\Delta^{p/4}(\psi(\Delta))^p]\Delta^{p/4}+\Delta^{p(1/p)}\Big)\\
&+2^{p-1}\bar{c}_{11}\int_{0}^{t_1}\mathbb{E}\Big(\sup_{0\leq t\leq s}|x_{\Delta}(t\wedge \vartheta_{\Delta,n})-x(t\wedge \vartheta_{\Delta,n})|^p\Big)ds\\
&\le \bar{c}_{10}(\bar{c}_1\Delta+o(\Delta))(\psi(\Delta))^p)+ \Big(\bar{c}_5\tau+2^{p-1}\bar{c}_{11}c_3\mathfrak{C}_1(\Delta^{p/4}(\psi(\Delta))^p+1)\Big)\Delta^{p(1/4  \wedge \Upsilon \wedge 1/p)}\\
&+2^{p-1}\bar{c}_{11}\int_{0}^{t_1}\mathbb{E}\Big(\sup_{0\leq t\leq s}|x_{\Delta}(t\wedge \vartheta_{\Delta,n})-x(t\wedge \vartheta_{\Delta,n})|^p\Big)ds\\
&\le \bar{c}_{10}(\bar{c}_1\Delta+o(\Delta))(\psi(\Delta))^p)+ \bar{c}_{12}\Delta^{p(1/4  \wedge \Upsilon \wedge 1/p)}\\
&+\bar{c}_{13}\int_{0}^{t_1}\mathbb{E}\Big(\sup_{0\leq t\leq s}|x_{\Delta}(t\wedge \vartheta_{\Delta,n})-x(t\wedge \vartheta_{\Delta,n})|^p\Big)ds,
\end{align*}
where $\bar{c}_{12}=\bar{c}_5\tau+2^{p}\bar{c}_{11}c_3\mathfrak{C}_1$ and $\bar{c}_{13}=2^{p-1}\bar{c}_{11}$. The Gronwall inequality gives us
\begin{align*}
\mathbb{E}\Big( \sup_{0\leq t \leq t_1}|x_{\Delta}(t \wedge \vartheta_{\Delta,n})-x(t \wedge\vartheta_{\Delta,n})|^p\Big)\le 
\mathcal{C}\Big((\Delta+o(\Delta))(\psi(\Delta))^p) \vee \Delta^{p(1/4\wedge \Upsilon \wedge 1/p)}\Big),
\end{align*}
as the required result in \eqref{eq:33}. where $\mathcal{C}= (\bar{c}_{10}(\bar{c}_1\vee 1)\vee \bar{c}_{12})e^{\bar{c}_{13}}$. By letting $\Delta\rightarrow 0$, we get \eqref{eq:34}.
\end{proof}
The strong convergence theorem of the truncated approximate solutions is as follows.
\begin{theorem}\label{eq:thrm2}
Let Assumptions \ref{eq:a1}, \ref{eq:a3*}, \ref{eq:a2} and \ref{eq:a3} hold. Then for any $p\ge 2$, we have
\begin{equation}\label{eq:42}
\lim_{\Delta\rightarrow 0}\mathbb{E}\Big( \sup_{0\leq t \leq T}|x_{\Delta}(t)-x(t)|^p \Big)=0
\end{equation}
and consequently
\begin{equation}\label{eq:43}
\lim_{\Delta\rightarrow 0}\mathbb{E}\Big( \sup_{0\leq t \leq T}|\bar{x}_{\Delta}(t)-x(t)|^p \Big)=0.
\end{equation}
\end{theorem}
\begin{proof}
Here, we only prove the theorem for $ p\ge 3$. As for $p\in [2,3)$, it follows directly from the case of $p=3$ and the H\"older inequality. Let $\tau_n$, $\varsigma_{\Delta,n}$ and $\vartheta_{\Delta,n}$, be the same as before. Set
\begin{equation*}
e_{\Delta}(t)=x_{\Delta}(t)-x(t).
\end{equation*}
For any arbitrarily $ \delta >0$, the Young inequality gives us
\begin{align}\label{eq:44}
\mathbb{E}\Big(\sup_{0\leq t \leq T}|e_{\Delta}(t)|^p\Big)&=\mathbb{E}\Big(\sup_{0\leq t \leq T}|e_{\Delta}(t)|^p1_{\{\tau_n>T \text{ and }\varsigma_{\Delta,n}>T\}}\Big)+\mathbb{E}\Big(\sup_{0\leq t \leq T}|e_{\Delta}(t)|^p1_{\{\tau_n\le T \text{ or }\varsigma_{\Delta,n}\le T\}}\Big)\nonumber\\
&\leq \mathbb{E}\Big(\sup_{0\leq t \leq T}|e_{\Delta}(t)|^p1_{\{\vartheta_{\Delta,n}>T\}}\Big)+\frac{\delta}{2}\mathbb{E}\Big(\sup_{0\leq t \leq T}|e_{\Delta}(t)|^{2p}\Big)\nonumber\\
&+\frac{1}{2\delta}\mathbb{P}(\tau_n\leq T \text{ or }\varsigma_{\Delta,n} \leq T).
\end{align}
So for $ p\ge 3$, Lemmas \ref{eq:L1} and \ref{eq:L7} give us
\begin{align}\label{eq:45}
 \mathbb{E}\Big(\sup_{0\leq t \leq T} |e_{\Delta}(t)|^{2p}\Big)&\leq 2^{2p}\mathbb{E}\Big(\sup_{0\leq t \leq T}|x(t)|^{2p}\vee \sup_{0\leq t \leq T}|x_{\Delta}(t)|^{2p}\Big)
 \nonumber\\
&\le 2^{2p}(c_1\vee c_3)^2.
\end{align}
By Lemmas \ref{eq:L0} and \ref{eq:L8},
\begin{equation}\label{eq:46}
\mathbb{P}(\vartheta_{\Delta,n}\leq T)\leq \mathbb{P}(\tau_n\leq T) +\mathbb{P}(\varsigma_{\Delta,n}\leq T).
\end{equation}
Also by Lemma \ref{eq:L9},
\begin{equation}\label{eq:47}
\mathbb{E}\Big(\sup_{0\leq t \leq T}|e_{\Delta}(t)|^p1_{\{\vartheta_{\Delta,n}>T\}}\Big)\leq \mathcal{C}\Big((\Delta+o(\Delta))(\psi(\Delta))^p) \vee \Delta^{p(1/4\wedge \Upsilon \wedge 1/p)}\Big).
\end{equation}
Substituting \eqref{eq:45}, \eqref{eq:46} and \eqref{eq:47} into \eqref{eq:44} yields
\begin{align*}
\mathbb{E}\Big(\sup_{0\leq t \leq T}|e_{\Delta}(t)|^p\Big)&\leq \frac{2^{2p}(c_1\vee c_3)^2\delta}{2}+\mathcal{C}\Big((\Delta+o(\Delta))(\psi(\Delta))^p) \vee \Delta^{p(1/4\wedge \Upsilon \wedge 1/p)}\Big)\\
&+\frac{1}{2\delta}\mathbb{P}(\tau_n\leq T)+\frac{1}{2\delta}\mathbb{P}(\varsigma_{\Delta,n}\leq T).
\end{align*}
Given $\epsilon\in (0,1)$, we can select $\delta$ so that
\begin{equation}\label{eq:48}
\frac{2^{2p}(c_1\vee c_3)^2\delta}{2} \leq\frac{\epsilon}{4}.
\end{equation}
Similarly, for any given $\epsilon\in (0,1)$, there exists $n_o$ so that for $n\geq n_o$, we may select $\delta$ to have
\begin{equation}\label{eq:49}
\frac{1}{2\delta}\mathbb{P}(\tau_n\leq T)\leq \frac{\epsilon}{4}
\end{equation}
and select $n(\epsilon)\leq n_o$ such that for $\Delta\in (0,\bar{\Delta}]$
\begin{equation}\label{eq:50}
\frac{1}{2\delta}\mathbb{P}(\varsigma_{\Delta,n}\leq T)\leq \frac{\epsilon}{4}.
\end{equation}
Finally, we may select $\Delta\in (0,\bar{\Delta}]$ sufficiently small for $\epsilon\in (0,1)$ such that
\begin{equation}\label{eq:51}
\mathcal{C}\Big((\Delta+o(\Delta))(\psi(\Delta))^p) \vee \Delta^{p(1/4\wedge \Upsilon \wedge 1/p)}\Big)\le \frac{\epsilon}{4}.
\end{equation}
Combining \eqref{eq:48}, \eqref{eq:49}, \eqref{eq:50} and \eqref{eq:51}, we get
\begin{align*}
\mathbb{E}\Big(\sup_{0\leq t \leq T}|x_{\Delta}(t)-x(t)|^p\Big)\leq \epsilon.
\end{align*}
as the required result in \eqref{eq:42}. By Lemma \ref{eq:L6}, we also get \eqref{eq:43} by setting $\Delta\rightarrow 0$.
\end{proof}
\newpage
\section{Numerical simulations}
Let us now implement the truncated EM (TEM) scheme for \textup{SDDE} \eqref{eq:2}. To illustrate the strong result established in Theorem \ref{eq:thrm2}, we compare the scheme with the backward EM (BEM) scheme. For justification regarding the choice of BEM scheme and its limitation, we refer the reader to consult \cite{emma}. Now consider the following form of \textup{SDDE} \eqref{eq:2}
\begin{align}\label{eq:s1}
dx(t)&=f(x(t^-),r(t))dt+\varphi(x((t-\tau)^-)),r(t))g(x(t^-))dB(t)+h(x(t^-),r(t))dN(t),
\end{align}
on $t\ge 0$ with initial values $\xi=0.02$ and $r_0=1$, where $r(t)$ is a Markovian chain defined on the state $\mathcal{S}=\{1,2\}$ with the generator given by
\begin{equation}
\Gamma=(\gamma)_{2\times 2}=
\begin{pmatrix}
-2&2\\
1&-1
\end{pmatrix}.
\end{equation}
Moreover, let
\begin{align}
f(x,i)&=
\begin{cases}
  0.3x^{-1}-0.2+0.1x-0.5x^{2},&\mbox{if $i=1$ }\\
  0.2x^{-1}-0.3+0.2x-0.6x^{2},& \mbox{if $i=2$,}
\end{cases}
\end{align}
\begin{align}
g(x)&=
\begin{cases}
  g(x)=x^{5/4}\\
  g(x)=x^{5/4}
\end{cases}
\end{align}
and
\begin{align}
h(x,i)&=
\begin{cases}
  h(x)=x,&\mbox{if $i=1$ }\\
  h(x)=2x,& \mbox{if $i=2$},
\end{cases}
\end{align}
$\forall(x,i)\in (\mathbb{R}\times \mathcal{S})$. The volatility process $\varphi(\cdot,\cdot)$ is a sigmoid-type function defined as follows:
\\\\
for $i=1$,
\begin{equation*}
\varphi(y,i)=
\begin{cases}
  \frac{1}{2}\frac{(1+(e^{y}-e^{-y}))}{(e^{y}+e^{-y})},&\mbox{if $y\geq 0$ }\\
  \frac{1}{4},&\mbox{Otherwise},
\end{cases}
\end{equation*}
and for $i=2$,
\begin{equation}
\varphi(y,i)=
\begin{cases}
  \frac{1}{4}\frac{(1+(e^{y}-e^{-y}))}{(e^{y}+e^{-y})}, & \mbox{if $y\geq 0$ }\\
  \frac{1}{8},                             & \mbox{Otherwise},
\end{cases}
\end{equation}
$\forall(y,i)\in (\mathbb{R}\times \mathcal{S})$. Obviously, all the assumptions imposed on $\varphi(\cdot,\cdot)$ are met(see \cite{emma}). We clearly see
\begin{equation*}
\sup_{1/u \leq x\leq u}(|f(x,i)|\vee g(x))\leq 3u^2,\quad u\ge 1,
\end{equation*}
We can now set $\mu=3u^2$ with inverse $\mu^{-1}(u)=(u/3)^{1/2}$.
\subsection{Numerical results}
By selecting $\psi(\Delta)=\Delta^{-2/3}$ and step size $10^{-3}$, we obtain Monte Carlo simulated sample path of $x(t)$ to \textup{SDDE} \eqref{eq:s1} at terminal time $T$ in Figure \ref{Fig:figure1} using the TEM scheme. The strong convergence between TEM and BEM numerical solutions is shown in Figure \ref{Fig:figure2}. In Figure \ref{Fig:figure3}, we observe the strong order to be approximately one half although this result is not yet proved theoretically. Do note that Figure \ref{Fig:figure2} and Figure \ref{Fig:figure3} were obtained without the $x^{-1}(t)$ drift term (see \cite{emma}). 
\begin{figure}[!htbp]
  \centerline{\includegraphics[scale=1]{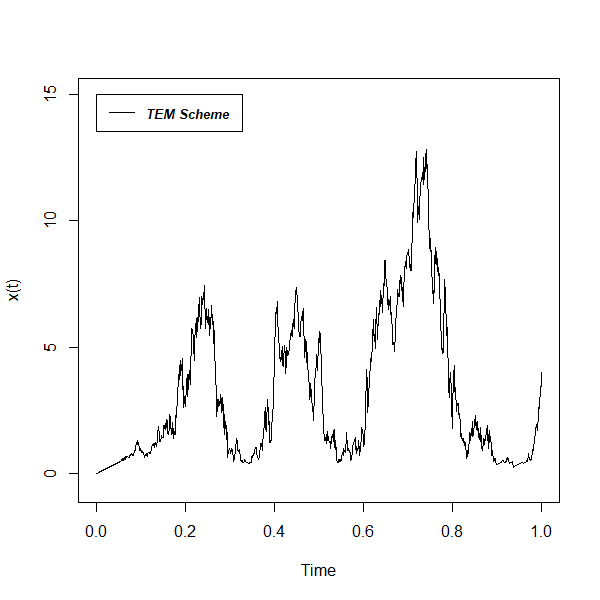}}
  \caption{Simulated sample path of $x(t)$ using $\Delta=0.001$}
  \label{Fig:figure1}
\end{figure}
\begin{figure}[!htbp]
  \centerline{\includegraphics[scale=1]{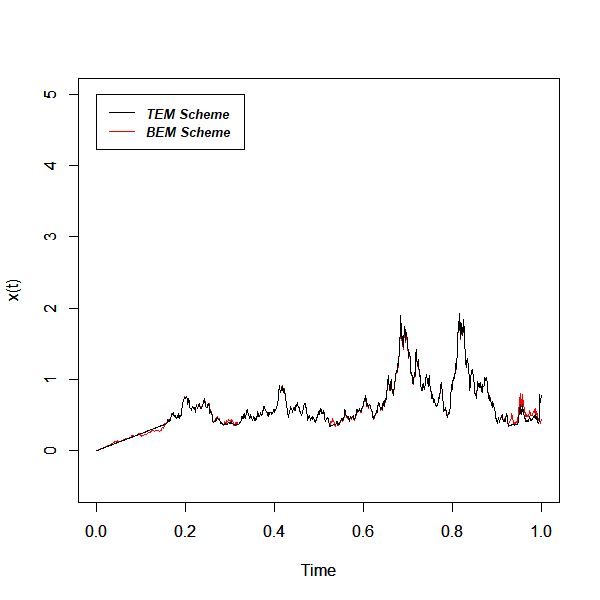}}
  \caption{Convergence of TEM and BEM solutions using $\Delta=0.001$}
  \label{Fig:figure2}
\end{figure}
\begin{figure}[!htbp]
  \centerline{\includegraphics[scale=1]{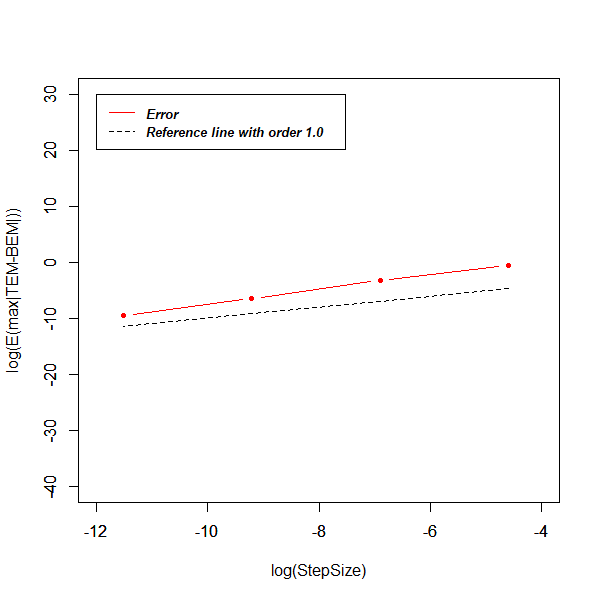}}
  \caption{Strong errors between TEM and BEM schemes}
  \label{Fig:figure3}
\end{figure}
\newpage
\section{Applications in finance}
In this section, we justify Theorem \ref{eq:thrm2} for Monte Carlo valuation of a bond and a barrier option.
\subsection{A bond}
\noindent Suppose the short-term interest rate is governed by \textup{SDDE} \eqref{eq:1}. Then a bond price $B$ at maturity time $T$ is given by
\begin{equation}\label{eq:52}
  B(T)=\mathbb{E}\Big[ \exp \Big(-\int_{0}^{T}x(t)dt \Big)\Big].
\end{equation}
The approximate value of \eqref{eq:52} using the step function \eqref{eq:22} is computed by
\begin{equation*}
  B_{\Delta}(T)=\mathbb{E}\Big[ \exp \Big(-\int_{0}^{T}\bar{x}_{\Delta}(t)dt \Big)\Big].
\end{equation*}
By Theorem \ref{eq:thrm2}, we get
\begin{equation*}
  \lim_{\Delta\rightarrow 0}|B_{\Delta}(T)-B(T)|=0.
\end{equation*}
\subsection{\textbf{A barrier option}}
Consider the payoff of a path-dependent barrier option at an expiry date $T$ defined by
\begin{equation}\label{eq:70}
  P(T)=\mathbb{E}\Big[ (x(T)-\Lambda)^+1_{\sup_{0\leq t\leq T }}x(t)<\mathbf{B})\Big].
\end{equation}
where the barrier, $\mathbf{B}$, and exercise price, $\Lambda$, are constants. Then the approximate value of \eqref{eq:70} using \eqref{eq:22} is computed by
\begin{equation*}
  P_{\Delta}(T)=\mathbb{E}\Big[ (\bar{x}_{\Delta}(T)-\Lambda)^+1_{\sup_{0\leq t\leq T }}\bar{x}_{\Delta}(t)<\mathbf{B})\Big].
\end{equation*}
So by Theorem \ref{eq:thrm2}, we also get
\begin{equation*}
  \lim_{\Delta\rightarrow 0}|P_{\Delta}(T)-P(T)|=0.
\end{equation*}
The reader is referred to \cite{highamao} for detailed coverage.
\section*{\textit{Acknowledgements}}
The author would like to express his sincere gratitude to his supervisor, Prof. Mao Xuerong and also thank University of Strathclyde for the doctoral scholarship. 

\end{document}